\documentclass[twocolumn]{svjour3}                     

\smartqed  

\usepackage{lmodern}
\usepackage{amsmath}
\usepackage{graphics}
\usepackage{epsfig}
\usepackage{epstopdf}
\usepackage[linesnumbered,ruled]{algorithm2e}
\usepackage{array}
\usepackage{stfloats}
\usepackage{url}
\usepackage{multirow}
\usepackage{pgfplots}
\pgfplotsset{width=6.4cm}\pgfplotsset{compat=newest}

\begin{document}

\title{A Clustering-Based Combinatorial Approach to Unsupervised Matching of Product Titles}


\author{Leonidas~Akritidis \and Athanasios~Fevgas \and Panayiotis~Bozanis \and Christos~Makris}


\institute{F. Author \at
              first address \\
              Tel.: +123-45-678910\\
              Fax: +123-45-678910\\
              \email{fauthor@example.com}           
           \and
           S. Author \at
              second address
}

\date{Received: date / Accepted: date}

\maketitle

\begin{abstract}
The constant growth of the e-commerce industry has rendered the problem of product retrieval
particularly important. As more enterprises move their activities on the Web, the volume and the
diversity of the product-related information increase quickly. These factors make it difficult for
the users to identify and compare the features of their desired products. Recent studies proved
that the standard similarity metrics cannot effectively identify identical products, since similar
titles often refer to different products and vice-versa. Other studies employed external data
sources (search engines) to enrich the titles; these solutions are rather impractical mainly
because the external data fetching is slow. In this paper we introduce UPM, an unsupervised
algorithm for matching products by their titles. UPM is independent of any external sources, since
it analyzes the titles and extracts combinations of words out of them. These combinations are
evaluated according to several criteria, and the most appropriate of them constitutes the cluster
where a product is classified into. UPM is also parameter-free, it avoids product pairwise
comparisons, and includes a post-processing verification stage which corrects the erroneous
matches. The experimental evaluation of UPM demonstrated its superiority against the
state-of-the-art approaches in terms of both efficiency and effectiveness.
\end{abstract}

\section{Introduction}
\label{sec-intro}

The online comparison of products is a crucial process, since it is usually the
first step in the life cycle of an electronic sale. Before a purchase is completed, the majority of
users search, collect and aggregate the characteristics both of the desired, and of any similar
products. For this reason, the role of the product comparison services has been rendered
increasingly important. These platforms retrieve data from various sources including electronic
stores, suppliers and reviews sites and they merge the information which refers to identical
products. In the sequel, they present this information to their users, allowing them to compare a
variety of parameters such as features and prices. They also facilitate the aggregation of user
opinions and reviews.

Since the products-related data originates from multiple sources, it presents a high degree of
diversity. To implement their comparison tools, the aggregation platforms must develop algorithms
which identify identical products. Apparently, the problem of product matching is vital for these
platforms, their users, and e-commerce industry in general.

Due to its importance, there exists a significant amount of research on this interesting problem.
The relevant literature includes solutions which can be divided into two categories: The first one
contains works which address the problem by examining solely the product titles. Earlier studies
employed standard string similarity methods including the cosine and edit-distance measures
\cite{sigmod2003,tkde2007,icde2007,irintro2008,tkdd2008,vldb2011,lcs2012,sigmod2013,ijca2013}. However,
\cite{cikm2012} showed that these metrics are inadequate on this particular problem; frequently,
identical products are described by very diverse titles, whereas highly similar titles do not
necessarily represent identical products.

For this reason, the method of \cite{cikm2012} employs Web search engines with the aim of enriching
the product titles with important missing words. For each title, the algorithm submits a query to a
Web search engine and, in the sequel, it collects and processes the returned results to identify
such words. A similar approach is introduced in \cite{vldb2014}, where the titles are modeled as
graphs and a clustering algorithm determines whether these graphs form a cohesive, or separately
clustered communities. However, the submission of a query in a search engine and the subsequent
processing of the returned results are expensive operations. Additionally, the provided APIs do not
allow unlimited usage and there is a limit to the number of the queries which can be submitted on a
daily basis. These limitations render these two approaches not applicable to large datasets with
millions of products.

The second category includes methods which take into consideration additional features such as
brands, manufacturers, categories, etc. More specifically, FEBRL provides an implementation based
on SVMs for learning suitable matcher combinations~\cite{febrl2008}, and MARLIN offers a set of
several learning methods such as SVMs and decision trees, combined with two similarity measures
\cite{sigkdd2003}. Nonetheless, these methods exhibit one significant problem: Since an aggregation
service is fed with data from multiple non controlled sources, many of the product attributes
which are present in one feed, may be absent in another. Even if an attribute is provided by all
sources, the data is frequently skewed or incomplete. In such occasions, it is inevitable that the
methods of this category will not perform well.


In this paper we present \textit{UPM (Unsupervised Product Matcher)}, an unsupervised algorithm
for matching products by their titles. The following list contains a brief description of the parts
of the algorithm and summarizes the contributions of this work:

\begin{itemize}
\item{UPM is based on the concept of unsupervised entity resolution via clustering. In details, it
constructs combinations of the words of the titles and assigns scores to each one of them,
similarly to \cite{inista2018}. The highest-scoring combination (called cluster) is the one which
best represents the identity of a product. All the products within the same cluster are considered
to be matching each other.}
\item{It performs morphological analysis of the product titles and identifies potentially useful
tokens (attributes, models, etc.). Each title is then split into virtual fields, and the tokens are
distributed to these fields according to their form and semantics.}
\item{It assigns scores to these fields and in the sequel, it plugs these scores into a function
which evaluates the combinations. This function also takes into consideration additional properties
of a combination, including its position in a title and its frequency.}
\item{It includes a post-processing verification stage which is executed after the formation of the
clusters. Based on the observation that very rarely a product appears twice within the catalog of a
vendor, this stage either moves products from one cluster to another, or it creates new clusters.
This stage leads to significant gains in the matching performance of the algorithm.}
\item{Unlike the aforementioned methods, our algorithm does not perform pairwise comparisons
between the products to determine whether they match or not. Therefore, it avoids the quadratic
complexity of this procedure, and also, it does not require the invention of an additional blocking
policy.}
\item{The following presentation introduces several parameters for UPM. However, there are global
settings for these parameters which consistently lead to satisfactory performance. The fixing
of these values ultimately leads to a method which is parameter-free.}
\end{itemize}

The rest of the paper is organized as follows: Section~\ref{sec-UPM} consists of six subsections
which describe the core parts of the algorithm. In particular, the first five present the
primary data structures and their construction method, the combinations scoring function, the
cluster selection strategy and the verification stage of UPM. Subsection \ref{ssec-tune} is
dedicated to the fixing of the various parameters. The experimental evaluation of the algorithm
is conducted in Section~\ref{sec-experiments} and the final conclusions are summarized in Section
\ref{sec-conclusion}. For research purposes, both the code we developed and the datasets we
utilized have been made publicly available on GitHub.

\section{Unsupervised Product Matching}
\label{sec-UPM}

Let us consider a set of vendors $V$ which includes electronic stores, suppliers, auction platforms
and so on. Each vendor $v \in V$ distributes an electronic catalog which contains the products s/he
provides, accompanied by some additional useful information. In case this information is organized
in a structured (or semi-structured) form, the catalog is called a \textit{feed} and the products
are stored as a collection of successive records. Each record is comprised of an arbitrary number
of attributes including its title, brand, model, and others.

Moreover, a vendor creates its feed independently of the others; hence, $v$ may provide information
about the brand or the category of a product, whereas $v'$ may not. Even if both $v$ and $v'$
include this information in their feeds, there may be discrepancies which inevitably lead to skewed
data.

Nevertheless, all feeds must contain at least one descriptive title for each included product. Two
or more vendors may use diverse titles to describe the same product. In the following subsections
we describe an unsupervised algorithm which matches products by overcoming this diversity.

\subsection{Combinations vs. n-grams}
\label{ssec-comb}

The string of a product title usually consists of multiple types of substrings, including words,
model descriptions, technical specifications, etc. We collectively refer to all these substrings as
\textit{tokens}. Let $W_t$ be the set of all tokens of a product title $t$. Then, a $k$-combination
$c_k$ is defined as any subset of $W_t$ of size $k$, without repetition and without care for tokens
ordering. For example, if $W_t$ consists of three tokens $W_t=\{w_1, w_2, w_3\}$, then there are
three possible 2-combinations, $\{w_1,w_2\},\{w_1,w_3\},\{w_2,w_3\}$, and one 3-combination,
$\{w_1,w_2,w_3\}$. In case $t$ consists of $l_t$ tokens (i.e., its length is $l_t$), then the
number of all possible $k$-combinations is equal to the binomial coefficient:

\begin{equation}
N(l_t,k)=\binom{l_t}{k}=\frac{l_t!}{k!(l_t-k)!}
\label{eq-num-comb}
\end{equation}

\noindent and the construction complexity $\forall k$ is exponential $O(2^{l_t})$.

Notice that $k$-combinations are different than $n$-grams: the latter are computed by sliding a
window of length $n$ over the examined string, from the left to the right; therefore, $n$-grams
capture only successive tokens. However, in a product title the important tokens (brand, model,
etc.) are usually scattered across the string and also, in non adjacent positions. Although the
construction of $k$-combinations is more expensive, they were preferred over $n$-grams because of
their ability to bring non-adjacent tokens together.

For example, there is no common 2-gram or 3-gram for the titles \textit{nVidia GeForce GTX1050 4GB}
and \textit{GeForce 4GB GTX1050}. Hence, $n$-grams cannot identify the similarity between these two
products. On the other hand, there are two common 2-combinations, namely \textit{GeForce GTX1050}
and \textit{GeForce 4GB}. Apparently, $k$-combinations outperform $n$-grams on this particular
problem.

Since the construction of all $k$-combinations is of exponential complexity, it is required to
limit their number to a minimum. Fortunately, our experiments showed that titles contain on average
6-11 tokens depending on the category of the product, and also, only a portion of them is important
for the identification of a product. For this reason, we limit the computations to the first
2-, 3-$,...,K$-combinations of the tokens of the involved titles.

Eventually, for a title which consists of $l_t$ tokens, the total number of combinations to be
computed is:

\begin{equation}
\Sigma N(l_t,K)=\sum_{k=2}^{k=K,k \leq {l_t}}\binom{l_t}{k}= \sum_{k=2}^{k=K,k \leq {l_t}} \frac{l_t!}{k!(l_t-k)!}
\label{eq-tnum-comb}
\end{equation}

For the sake of simplicity, in the presentation which follows we use the term ``combination''
instead of $k$-combination, and the simplified notation $c$ instead of $c_k$.

\subsection{Morphological Analysis \& Token Semantics}
\label{ssec-moran}

Each title consists of tokens which are not equally important for the description of a product.
Vendors may provide irrelevant information in a title, including payment facilities, special
discounts, offers, shipping and delivery data, availabilities and so on. Such kinds of information
are considered as noise; consequently, they may degrade the effectiveness of an entity matching
algorithm.

The unsupervised extraction of the hot tokens from a title is a particularly challenging task,
since vendors use different syntactical rules to express the information of their products and,
also, each product type presents its own specificity. Nevertheless, in this paper we perform
morphological analysis of the titles with the aim of identifying these hot tokens. In particular,
we initially examine the form of the tokens and we categorize each one of them as:

\begin{itemize}
\item{\textit{Mixed}, in case it contains both digits and letters, or}
\item{\textit{Numeric}, in case it contains only digits (with a thousands or a decimal separator),
or}
\item{\textit{Alphabetic}, in all other cases.}
\end{itemize}

In the sequel, we identify the following important pieces of information:

\textit{1) Product Attributes}: The attributes of a product are important, since they can be used
to differentiate it from another similar product. For instance, the \textit{32 GB} version of a
cell phone is a different product compared to the \textit{64 GB} version of the same model. This is
valid for multiple product types (e.g. hardware, electrical and electronic devices, etc.).

The process is based on a small lexicon of measurement units (e.g. bytes, hz, bps, meters, etc.)
and of their multiples and sub-multiples. By employing this lexicon, an attribute is identified
either i) when a pair of a numeric token and a measurement unit is encountered (e.g. \textit{32
GB}), or ii) when the ending of a mixed token is a measurement unit and its suffix consists
of digits only (e.g. \textit{32GB}). In the former case, the two tokens of the pair are concatenated
into one, with the aim of eliminating the difference with the latter case.

\textit{2) Models}: The model descriptors are the most fundamental part of a product title, since
they represent its identity. Unfortunately, the models may receive forms which vary significantly
among vendors, and moreover, a specific model may appear under different forms (e.g. \textit{PS 3}
vs. \textit{PS3} vs. \textit{Playstation3}, etc.). Consequently, it is particularly hard for an
unsupervised technique to correctly identify such model descriptors with absolute accuracy.

Nevertheless, the approach we present here yields significant improvements in the performance of
our matching algorithm. We consider that a token is a possible model descriptor if it is either
mixed, or numeric and it is not followed by a measurement unit. In addition, not all mixed tokens
are treated equally. For instance, the first mixed token in a product title is considered to be
more possible to contain a model compared to the second or the third mixed token.

In case a token does not fall into one of the above categories, then it is classified as a
\textit{normal} token. Table \ref{tab-toksem} summarizes the five aforementioned semantics
accompanied by their identification rules.

\begin{table}[t]
\begin{center}
\smaller
\caption{Identification rules of the token semantics}
\vspace{-0.2cm}
\label{tab-toksem}
\begin{tabular}{|p{0.6cm}|p{1.2cm}|p{5.7cm}|} \hline
{\bf Type} & {\bf Semantics} & {\bf Identification Rule/s}                                                                        \\\hline
     1     & Attribute       & i) numeric tokens followed by measurement units, or ii) mixed tokens  ending in a measurement unit \\\hline
     2     & Model           & The first mixed token in the title which does not represent an attribute                           \\\hline
     3     & Model           & All the rest mixed tokens in the title which do not represent an attribute                         \\\hline
     4     & Model           & A numeric token which is not followed by a measurement unit                                        \\\hline
     5     & Normal          & All the other tokens of the title                                                                  \\\hline
\end{tabular}
\end{center}
\vspace{-0.3cm}
\end{table}

The morphological analysis of a title includes several additional steps which are performed with
the aim of removing the discrepancies between tokens with the same meaning. More specifically, the
product titles of the dataset are parsed sequentially and the following procedures are applied to
the extracted tokens:

\begin{itemize}
\item{Case folding: all letters are converted to lower case.}
\item{Punctuation removal: all punctuation symbols and marks are removed from a title apart from
i) dots and commas which are thousands or decimal separators, and ii) hyphens and slashes which
delimit tokens. In the latter case, these tokens are appended in the title.}
\item{Duplicate tokens removal: the existence of two or more identical tokens in a product title is
rare. However, we found that their removal improves the performance of the algorithm by a
significant margin.}
\end{itemize}

\subsection{Data Structures Construction}
\label{ssec-ds}

After the tokenization and the morphological analysis of a title has been completed, the extracted
tokens and combinations are used to build the following data structures:

\subsubsection{Tokens Lexicon}
\label{sssec-toklex}

This is an ordinary lexicon structure $L_w$, which is used to store the tokens extracted from the
titles of the products. For each token $w$, the tokens lexicon also maintains:

\begin{itemize}
\item[i)] a unique integer identifier (token ID),
\item[ii)] a frequency value $f_w$ which represents the number of products that contain $w$
in their titles, and
\item[iii)] a special variable $s_w$ which is set equal to the semantics of $w$, as indicated by
the first column of Table~\ref{tab-toksem}.
\end{itemize}

\subsubsection{Combinations Lexicon}
\label{sssec-comblex}

The combinations lexicon $L_c$ stores the $k$-combinations (${k}\leq{K}$) of the tokens of the
product titles. The representation of the stored combinations is of particular importance, since
it must support not only fast searching, but also searching for combinations with different
orderings of their tokens. For instance, consider the case where we extract the 3-combination
\textit{CPU 3.2GHz 32MB}, which does not exist in $L_c$. Instead, suppose that $L_c$ contains the
3-combination \textit{CPU 32MB 3.2GHz}, which is the same as the one we search for, but with
different ordering of its tokens. In such cases, we desire to identify the equality between the two
records to avoid the insertion of the same combination twice.

\begin{figure}[t]
\centering
\epsfig{file=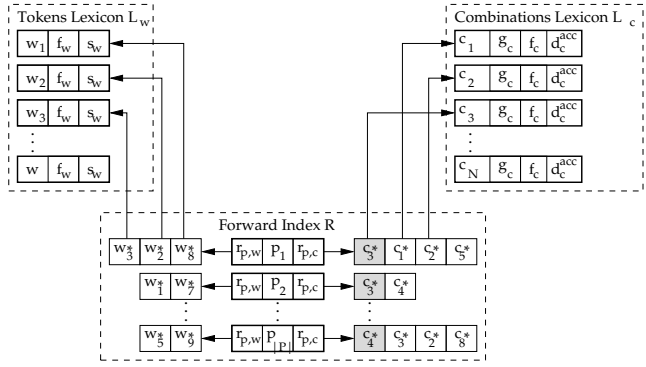, width=3.3in}
\caption{The connection of the forward index $R$ with the tokens lexicon $L_w$ (left), and the
combinations lexicon $L_c$ (right). The grayed boxes indicate the highest-scoring combination (that
is, the cluster) of each product $p \in R$. In this example, $p_1$ and $p_2$ match each other,
because they belong to the same cluster $c_3$. }
\label{fig-ds}
\vspace{-0.1cm}
\end{figure}

The proposed algorithm satisfies this requirement by assigning signatures to all combinations.
The key concept is that \textit{a combination must have the same signature independently of the
ordering of its component tokens}. This means that the two 3-combinations of the previous example,
i.e. \textit{CPU 32MB 3.2GHz} and \textit{CPU 3.2GHz 32MB} must be assigned equal signatures. More
specifically, the following procedure is applied: before a combination $c$ is inserted into $L_c$,
its signature $g_c$ is computed. In case $g_c$ is not found in $L_c$, then $c$ does not exist in
$L_c$ in any form (i.e. under any ordering of its tokens) and can safely be inserted to it.
Otherwise, $c$ resides within $L_c$ in one form or another.

A simple method for computing the signature of a combination $c$ is via tokens sorting and hashing.
More specifically, this method initially retrieves the IDs of the component tokens of $c$ and sorts
them in increasing order. The sorted values are then concatenated and delimited by a special symbol
(e.g. a single space character). The string we obtain is subsequently passed through a string hash
function $h$ which is common for all combinations. The output of $h$ constitutes the desired
signature $g_c$.

As we demonstrate later in the experimental evaluation, the usage of signatures leads to
substantial improvements in the efficiency of the algorithm. Eventually, each combination record
$c \in L_c$, possesses the following properties:

\begin{itemize}
\item[i)] its signature $g_c$,
\item[ii)] a frequency value $f_c$ which represents the number of product titles which contain
$c$, and
\item[iii)] a distance accumulator $d^{acc}_c$ which maintains the sum of the distances of $c$ from
the beginning of the titles. This value will be used later to assign a score to $c$.
\end{itemize}

\subsubsection{Forward Index}
\label{sssec-fi}

The forward index $R$ is essentially a list of all product records. Each product $p$ is associated
with two pointer lists:

\begin{itemize}
\item[i)] the tokens forward list $r_{p,w}$, which maintains $l_t$ pointers to the tokens of the
title $t$ of $p$; and
\item[ii)] the combinations forward list $r_{p,c}$, namely, a list of $\Sigma N(l_t,K)$ pointers
(given by eq.~\ref{eq-tnum-comb}). Each pointer refers to a combination $c$ of $p$, where
$c \in L_c$.
\end{itemize}

In Figure~\ref{fig-ds} we depict the interconnection of the forward index with the tokens and the
combinations lexicon structures. Notice that the existence of pointers in the forward index saves
us the cost of storing the same data twice.

\subsubsection{Construction Algorithm}
\label{sssec-dscon}

\begin{algorithm}[t]
\small
\label{algo-1}
    \ {initialize lexicons $L_w$ and $L_c$ and the forward index $R$}\;
    \ {create the measurement units table $M$}\;
    \For {each product $p$ with title $t$} {
        $R$.insert$(p)$\;
        $W_t \leftarrow$ tokenize the title $t$ of $p$\;
        \For {each token $w \in W_t$} {
			$s_w \xleftarrow{M}$ rules of Table~\ref{tab-toksem}, Subsection \ref{ssec-moran}\;
			$w* \leftarrow L_w.$search$(w)$\;
            \eIf{$w*>0$} {
                $f_w \leftarrow f_w + 1$\;
            } {
                $f_w \leftarrow 1$\;
                $w* \leftarrow L_w.$insert$(w)$\;
            }
            $R$.insert$(p,w*)$\;
        }

        \For {each $k\in[2,K]$} {
            compute all $k$-combinations $C_k$ of $t$\;
            \For {each $k$-combination $c \in C_k$} {
                $g_c \xleftarrow{h}$ compute signature of $c$\;
                $c* \leftarrow L_c.$search$(g_c)$\;
                $d(c,t) \leftarrow$ compute the distance of $c$ from $t$\;
                \eIf{$c*>0$} {
                    $f_c \leftarrow f_c + 1$\;
                    $d^{acc}_c \leftarrow d^{acc}_c + d(c,t)$\;
                } {
                    $f_c \leftarrow 1$\;
                    $d^{acc}_c \leftarrow d(c,t)$\;
                    $c* \leftarrow L_c.$insert$(c)$\;
                }
                $R$.insert$(p,c*)$\;
            }
        }
    }
    \caption{Data structures construction}
\end{algorithm}

Algorithm~\ref{algo-1} presents the construction methodology of the aforementioned data structures.
Initially, each product $p$ enters the forward index $R$ with its tokens and combinations
lists empty (step 4). In the sequel, its title $t$ is parsed and its tokens $W_t$ are extracted.
Each token $w \in W_t$ passes through a filtration process where the morphological analysis of the
previous subsection is performed (i.e. case folding, punctuation removal, etc.). Moreover, the
semantics $s_w$ of $w$ is identified according to the rules of Table \ref{tab-toksem} (step 7).

After this process, a search for $w$ in $L_w$ is performed (step 8). Notice that the search
operation returns a pointer $w*$ to the corresponding token record in $L_w$. In case the search is
unsuccessful, $w$ is inserted in $L_w$ with $f_w=1$ and a pointer to the new record is returned; in
the opposite case, its corresponding frequency value $f_w$ increases by 1 (steps 9--14). Finally,
the pointer $w*$ is inserted into the tokens list of $p$ within the forward index $R$ (step 15).

The procedure continues with the computation of all $2,3,...,K$-combinations of $t$ and the
generation of their respective signatures (steps 17--20). Then, for each combination $c$ the
lexicon $L_c$ is queried against its signature $g_c$. If this search is unsuccessful, $c$ is
inserted in $L_c$ with $f_c=1$ and a pointer $c*$ to the new record is returned; otherwise, $f_c$
increases by one (steps 23--30). The algorithm ends with the insertion of the pointer $c*$ in the
combinations list of $p$ within the forward index $R$ in step 31.

During this process, the distance $d(c,t)$ of $c$ from the beginning of the product title $t$ is
calculated, and it is used to update the distance accumulator $d^{acc}_c$. This distance value will
be employed later by the combinations scoring function. We provide more details about the usage of
the distance accumulator in the next subsection.

\subsection{Scores Computation \& Cluster Selection}
\label{ssec-phase2}

In summary, the purpose of this phase is to compute an importance score $I_{c}$ for each
combination $c \in r_{p,c}$ of each product $p$ of the forward index. The highest-scoring
combination will then be declared as the dominating cluster $u$ where $p$ will be mapped to. All
the other products which will also be mapped to $u$ will be considered that they match $p$.
Finally, the clusters of all products will be utilized to build the clusters universe $U$ which
shall assist us further.

We now elaborate on the form of the combinations score function. Initially, we study the properties
that a combination must possess to be declared as a dominating cluster, and then we proceed to the
quantification of these properties.

\begin{itemize}
\item{\textit{Frequency}: The number of products which contain $c$ is an important parameter, since
the more frequent a combination is, the more products will be mapped to it. In contrast, if we
select a rare combination, we shall not be able to map any other product to it.}
\item{\textit{Length}: The frequency criterion definitely favors the short combinations because it
is more possible to encounter a 2-combination which is common for multiple products, compared to a
3-combination. However, the short combinations are not as descriptive as the longer ones and also,
there is a risk of creating very inhomogeneous clusters which may erroneously contain different
products.}
\item{\textit{Position}: A broadly accepted idea in information retrieval dictates that the most
important words of a document usually appear early, that is, in a small distance from its
beginning.}
\item{\textit{Hot tokens}: A combination which contains multiple highly informational tokens
represents the identity of the product more accurately compared to one which does not include such
tokens.}
\end{itemize}

Given a title $t$, a combination $c$ of $t$, and a token $w\in c$, we consider that $o_w^{(c)}$
is the position (or offset) of $w$ in $c$ and $o_w^{(t)}$ is the position of $w$ in $t$. By using
this notation, the distance $d(c,t)$ between $c$ and $t$ is computed by employing the
well-established Euclidean distance for strings:

\begin{equation}
d^2(c,t)=\sum_{w\in c} { \big(o_w^{(c)} - o_w^{(t)}\big)^2}
\label{eq-dist}
\end{equation}

Based on this equation we compute the average distance of $c$ from the beginning of all titles as
follows:

\begin{equation}
\overline{d(c)}=\frac{1}{f_c}\sum_{\forall{t \supset c}}{d(c,t)}
\label{eq-avgdist}
\end{equation}

The four aforementioned properties of a combination can now be quantified by the following scoring
function:

\begin{equation}
I(c) = \frac{Y_c^2}{\alpha + \overline{d(c)}} \log{f_c}
\label{eq-score}
\end{equation}

\noindent
where $\alpha>0$ is a constant which i) prevents $I(c)$ from getting infinite when $\overline{d(c)}
=0$ (i.e, when $c$ appears always in the beginning of all titles), and ii) determines the
importance of proximity in the overall score of a combination.

The $Y_c$ factor constitutes the IR score of $c$ and it is built by adopting the spirit of the
BM25F scoring method for structured documents \cite{iexml2005}. This scheme is designed to boost
the scores of the words which appear in highly important places of a document (called
\textit{fields}), such as its title.

Although a product title is clearly a short unstructured text, here we introduce the idea of
splitting a title into \textit{virtual fields}, based on the aforementioned semantics of each
token. According to this approach, a title is divided into five virtual fields, from $z_1$ to
$z_5$. Each field is allowed to contain only tokens which have identical semantics. For example,
according to Table~\ref{tab-toksem}, $z_1$ shall accommodate only the tokens which represent the
attributes of a product, whereas $z_2, z_3$ and $z_4$ enlist the tokens which potentially carry
information about the model. Notice that a field may be completely empty, whereas a token can
belong to only one field.

Similarly to BM25F, the $Y_c$ factor is computed by applying the following equation:

\begin{equation}
Y_c = \sum_{\forall w \in c} { idf(w) \frac{ Q(z_{s_w}) }{1-b+bk/\overline{l_c}} }
\label{eq-ir}
\end{equation}

\noindent
where $Q(z_{s_w})$ is the weight of the field which contains a token $w \in c$. Notice here the
dependence of this weight from the semantics value $s_w$. Furthermore, $idf(w)=\log{(|P| / f_w)}$
is the inverse document frequency of $w$ (where $|P|$ is the total number of product titles).
Recall also that $k$ symbolizes the length of $c$, whereas $\overline{l_c}$ is the average length
(in number of tokens) of all combinations in the dataset. Finally, $b$ is a constant whose value
falls into the range $[0,1]$.

In conclusion, eq.~\ref{eq-score} indicates that a product should be clustered under a combination
which: i) is frequent, ii) is reasonably long, iii) usually occurs near the beginning of the titles
and iv) contains multiple important tokens. In the sequel, we employ it to identify the most
appropriate cluster for every product of the forward index $R$.

\begin{algorithm}[t]
\small
\label{algo-2}
	initialize the clusters universe $U$\;
    \For {each product $p$ in $R$} {
        $r_{p,c} \leftarrow $retrieve the combinations forward list of $p$\;
        $\max \leftarrow 0$;~~~~$u \leftarrow$ NULL\;
        \For {each combination $c \in r_{p,c}$} {
		    \For {each token $w \in W_c$} {
		    	$X[s_w] \leftarrow X[s_w] + 1$\;
	    	}
			$Y_c \leftarrow 0$\;
   	        \For {each token $w \in W_c$} {
       			$Q(z_{s_w}) \leftarrow$ Eq. \ref{eq-q}\;
   	        	$Y_c \leftarrow Y_c + $ Eq. \ref{eq-ir}\;
   	        }
			$\overline{d(c)} \leftarrow d^{acc}_c / f_c$\;
			$I_{c} \leftarrow$ Eq. \ref{eq-score}\;
			\If{$I_{c} > \max$} {
				$\max \leftarrow I_{c}$\;
				$u \leftarrow c$\;
			}
        }
		$U$.insert$(u,p)$;~~~~// Algorithm \ref{algo-3}\
    }
    deallocate $L_c-U$\;
    \caption{Score calculation and cluster selection.}
\end{algorithm}

Algorithm \ref{algo-2} contains the details of this procedure. Notice that since we are only
interested in the highest-scoring combination, it is not mandatory to store the scores of all
combinations in some dedicated data structure (e.g. heap); a simple computation of the maximum
score suffices.

Initially, an empty set $U$ is initialized. In the sequel, we iterate through the products of $R$
and for each product ${p}\in{R}$ we traverse its combinations forward list $r_{p,c}$. For each
combination $c \in r_{p,c}$, the field lengths are stored within an array $X$, according to the
semantics of the tokens of $c$ (steps 6--8). In steps 9--13 the IR score of eq. \ref{eq-ir} is
calculated, whereas the next step computes the average distance $\overline{d(c)}$. Having prepared
this data, the score of $c$ is obtained in step 15. In steps 16--19 we conditionally update the
maximum score and the highest-scoring combination.

The combination with the maximum score is subsequently selected as the \textit{dominating cluster},
or simply the \textit{cluster} $u$ of $p$. In the sequel, $u$ is inserted into the global set $U$,
along with the corresponding product $p$, according to the steps 2--6 of Algorithm \ref{algo-3}.
Notice that insertion includes additional operations after step 6, which are described in details
in the next subsection. Finally, the algorithm deallocates the resources occupied by data which are
not useful for the next step, including the combinations which have not been declared clusters,
that is, $L_c-U$.

\subsection{Verification Stage \& Cluster Correction}
\label{ssec-phase3}

The procedures of the previous subsections achieve their goal, that is, unsupervised product
matching by using only their titles. However, there is still room for improvement.

Here we present a post-processing verification step which attempts to recognize false matches. In
the absence of training data, it is based on a simple, but strong hypothesis: In the vast majority
of cases, \textit{each product appears only once in the feed of the same vendor}, or equivalently,
\textit{a vendor does not include identical products in his/her catalog}. Of course, there are some
individual cases where the same product indeed exists multiple times within a catalog of a vendor.
However, such cases are extremely rare and they usually occur by mistake.

This hypothesis, combined with the fact that a cluster contains products which are considered to
match each other (i.e. they are identical), leads to the following lemma:

\begin{lemma}
A cluster $u$ cannot contain two or more products from the same vendor $v$.
\end{lemma}

\begin{proof}
Suppose that $u$ contains two products $p_1$ and $p_2$ from the same vendor $v$. Since $u$ contains
only products which match each other, $p_1$ is identical to $p_2$. But then $v$ included the same
product multiple times in his/her catalog, a statement which contradicts our hypothesis.
\end{proof}

This lemma drives the entire verification stage. Based on it, we say
that $v$ is a \textit{violator} of $u$, if $u$ contains two or more of his/her products. In this
case, $u$ is an \textit{invalid cluster} and it requires a special validation process to be applied
to it. In short, this process: i) allows only one product of $v$ in $u$, and ii) evicts the rest
products of $v$ from $u$. The evicted products can either: i) migrate to another existing cluster
according to some criteria, or ii) be transferred to a new cluster.

Recall that technically, a cluster is merely a combination object and as such, it posesses the
properties of Subsection \ref{sssec-comblex}. To support the verification stage, a cluster $u$ must
be extended with the following elements:

\begin{itemize}
\item{A list $V_u$ with the vendors of the products of $u$,}
\item{One of the products of $u$ is selected as the \textit{representative product $\pi_u$} of $u$,
according to a score. The title of $\pi_u$ is used as a label for $u$ and thus, cannot leave $u$.}
\item{One list $P_{u,v}$ per vendor $v \in V_u$ which stores the products that both belong to
$u$ and are provided by $v$. Each product $p \in P_{u,v}$ is assigned two scores: i) $S^1_p=\sum_{w \in t}{idf(w)}$
 will be used to select the representative product $\pi_u$, and ii)
$S^2_p$ which stores the similarity of $p$ with $\pi_u$.}
\end{itemize}

These elements are computed immediately during the insertion of $u$ and $p$ into $U$ (step 21 of
Algorithm \ref{algo-2}). The steps 7--14 of Algorithm \ref{algo-3} describe this process:
initially, the vendor $v$ of $p$ is inserted into the list $V_u$ (provided that $v \notin V_u$). In
the sequel, the score $S^1_p$ is computed, and in case it exceeds the maximum product score in the
cluster, then $p$ is declared as the representative product $\pi_u$ of the cluster.

\begin{algorithm}[t]
\small
\label{algo-3}
\SetAlgoLined
\SetKwProg{Fn}{Function}{ }{end}
	\Fn{$U$.insert$(u,p)$} {
        $v\leftarrow$ vendor of $p$\;
        \If {${u}\notin{U}$} {
            $U$.append$(u)$\;
        }
   		$P_{u,v}$.insert$(p*)$\;
        \If {${v}\notin{V_u}$} {
		  $V_u$.insert$(v)$\;
        }
		$S^1_p \leftarrow \sum_{w \in t}{idf(w)}$\;
		\If {$S^1_p > \max{S^1_p}$} {
			$\max{S^1_p} \leftarrow S^1_p$\;
			$\pi_u \leftarrow p$\;
		}
	}
    \caption{Insertion of a cluster $u$ and a product $p$ into the universe $U$.}
\end{algorithm}

After the required data has been prepared, the verification stage of Algorithm \ref{algo-4} is
executed. For each cluster $u \in U$ we traverse its list of vendors $V_u$ and in case a violator
$v$ is found (i.e., $|P_{u,v}| > 1$), we identify which product of $v$ will stay in $u$. This is
achieved by calculating the similarity score of each product $p \in P_{u,v}$ with the
representative product $\pi_u$, and by sorting $P_{u,v}$ in decreasing similarity score order
(steps 3--7). The first record of the list, namely, the most similar product to $\pi_u$, is
selected to remain in $u$; the rest $(P_{u,v}-P_{u,v}[0])$ products will eventually abandon $u$.

There exist two options to handle the evicted products. The former is applied when there exists
another cluster $u' \in U$ whose representative product $\pi_{u'}$ is highly similar to an evicted
product $p$. In that case, $p$ migrates to $u'$, provided that $u'$ does not contain any other
product of $v$ and it will not become invalid after the insertion of $p$. If no cluster of $U$
satisfies this criteria, then the latter option dictates that we create a new cluster $u''$, append
$u''$ to the universe $U$, and finally, transfer $p$ to $u''$ (steps 8--17).

The final point which needs to be clarified is the method for retrieving the clusters which are
both valid and relevant to an evicted product $p$ of a vendor $v$ (step 10). The strategy we
adopted was to compute the cosine similarity of $p$ with the representative product of each
candidate cluster which did not contain any other products of $v$. In case the maximum computed
similarity is above a predefined threshold $\tau$, then $p$ is inserted into the corresponding
cluster. Otherwise, a new cluster is created and $p$ is transferred there.

\subsection{Parameter Fixing}
\label{ssec-tune}

\begin{algorithm}[t]
\small
\label{algo-4}
    \For {each cluster $u$ in $U$} {
	    \For {each vendor $v$ in $V_u$} {
	    	\If {$|P_{u,v}| > 1$} {
	    		\For {each product $p \in P_{u,v}$} {
	    			compute $S^2_p \leftarrow sim(p,\pi_u)$\;
	    		}
	    		sort $P_{u,v}$ in decreasing $S_p$ order\;
	    		\For {each product $p \in (P_{u,v}-P_{u,v}[0])$} {
	    			evict $p$ from $u$\;
	    			$u' \leftarrow$ retrieve cluster // apply criteria\;
	    			\eIf {$S^2_{\pi_{u'}} > \tau$} {
	    				$u'$.insert$(p)$\;
	    			} {
	    				create new cluster $u''$\;
	    				$U$.insert$(u'',p)$;~~~~// Algorithm \ref{algo-3}\
	    			}
	    		}
	    	}
	    }
    }
    \caption{Verification and cluster reselection.}
\end{algorithm}

Until this point, we introduced five parameters in the presentation of UPM. Here we fix  the
values of these parameters based on the conclusions of exhaustive experimentation with multiple
datasets. The purpose of setting fixed values to all parameters is to present an algorithm which is
not only unsupervised, but also parameter-free.

We begin with $K$, the modifier which determines the maximum number of tokens which can be used in
a single combination. In all cases, the value $K^{*}$ which maximized the effectiveness of the
algorithm was found to be equal to the half of the average title length, that is:

\begin{equation}
K^{*}=\left \lfloor { \overline{l_t} / 2 } \right \rfloor
\label{eq-K}
\end{equation}

Larger or smaller values of $K$ have a negative impact on performance. This observation leads to
the conclusion that, on average, only a portion of the tokens of a title are actually important for
the identification of a product. This conclusion established the basis of UPM+, a simple variant
which takes into consideration only the first $2K^{*}$ tokens of a title, and ignores the rest of
them. Therefore, the extracted combinations are reduced by a significant margin (especially in the
case of long titles), whereas it is anticipated that we only suffer a small loss in matching
performance. This anticipation is verified experimentally in Section \ref{sec-experiments}.

In addition, Eq.~\ref{eq-score} depends on $\alpha$, which determines the
importance of proximity in the score of a combination. Our experiments revealed that the setting
$\alpha=1$ maximized the effectiveness of UPM in all examined cases.

The third parameter of the algorithm is $b$, and it was introduced in Eq.~\ref{eq-ir}. The
value of $b$ which consistently led to satisfactory results was $b=1$.

The next parameter to determine is the field weights of Eq.~\ref{eq-ir}. The simplest solution
here is to assign a fixed weight value to each field; for instance, one may consider that the model
fields are twice as important as the field which contains the normal tokens. Although this approach
delivers good results in some cases, it has two problems: i) the weights are set arbitrarily in an
ad-hoc manner, and ii) a set of predefined field weights which works well in one case, may lead to
poor performance in another.

For these reasons, we dropped the idea of assigning fixed values to the field weights. Instead, we
discovered a function which leads to satisfactory performance in all cases:

\begin{equation}
Q(z_{s_w}) = \frac{|W|}{X[s_w]}
\label{eq-q}
\end{equation}

\noindent
where $|W|$ is the total number of the distinct tokens of the product titles, and $X$ is an array
with a size equal to five. Each entry in $X$ represents the number of tokens of the field which is
associated to its index. For instance, in conjunction with the first column of Table
\ref{tab-toksem}, $X[1]$ stores the population of $z_1$, that is, the number of tokens in the title
which represent an attribute of the product. Equation~\ref{eq-q} implements the intuition that the
more tokens a field contains, the less important its tokens are, and vice versa.

Finally, we determine the value of the parameter $\tau$ of Algorithm \ref{algo-4}. Recall that this
parameter controls the similarity threshold of an evicted product with a candidate cluster. The
value which maximized performance was $\tau=0.4$.

\section{Experiments}
\label{sec-experiments}

This section analyzes the results of the experimental evaluation of the proposed algorithm. In
particular, we compare UPM and UPM+ with two popular string similarity metrics, i.e, cosine
similarity, and Jaccard index, as well as their enhanced versions, which include IDF token weights.
Given two titles $t$ and $t'$, these metrics are defined as follows:

\begin{itemize}
\item{cosine similarity: $\mathcal{CS}={|{t}\cap{t'}|}/{\sqrt{|t|}\sqrt{|t'|}}$,}
\item{cosine similarity with IDF token weights:
\begin{equation}
\nonumber
\mathcal{CS'}=\frac{\sum_{w\in(t\cap{t'})}{idf_w^2}}{\sqrt{\sum_{w\in t}{idf_w^2}}\sqrt{\sum_{w\in t'}{idf_w^2}}}
\label{eq-cosidf}
\end{equation}
}
\item{Jaccard index: $\mathcal{J}={|{t}\cap{t'}|}/{|{t}\cup{t'}|}$, and}
\item{Jaccard index with IDF token weights:
\begin{equation}
\nonumber
\mathcal{J'}=\frac{\sum_{w\in(t\cap{t'})}{idf_w^2}}{\sum_{w\in(t\cup{t'})}{idf_w^2}}
\label{eq-jacidf}
\end{equation}
}
\end{itemize}

To ensure the robustness of our evaluation and to avoid results which were accidentally obtained,
we based our experiments on multiple datasets. In particular, we crawled two popular product
comparison platforms, PriceRunner\footnote{https://www.pricerunner.com/} and
Skroutz\footnote{https://www.skroutz.gr/}, and we constructed 8 datasets out of each one. Each of
these 16 datasets represents a specific product category. The categories were selected with two
criteria: i) to study the performance difference of the same methods on similar products which were
provided by different vendors, and ii) to examine the effectiveness of the algorithms on products
from diverse categories. For this reason, we include products from both identical and different
categories in our experiments. Moreover, we created one aggregate dataset per platform, which
contains all the products from all 8 categories combined. These datasets enable the examination of
the performance on heterogeneous datasets.

To facilitate prices and features comparison, the platforms group the same products into clusters.
These clusters were utilized to establish the ground-truth for the evaluation of the various
methods. More specifically, similarly to UPM, both platforms consider that all the titles within
a cluster represent the same product. Hence, each dataset is accompanied by a special ``matches''
file, which stores all the pairs of matching titles of all clusters. This file is subsequently used
to verify the effectiveness of each method.

Table \ref{tab-data} presents the 18 experimental datasets accompanied by several useful
characteristics. The first 9 rows concern the datasets which were crawled from PriceRunner, whereas
the next 9 are about the ones which were acquired from Skroutz. The columns 2, 3, and 4 display the
distinct number of vendors, products, and product titles of each dataset respectively. Moreover,
the fifth column shows the average length of the titles; this parameter is important because it
determines the value of $K$ according to eq. \ref{eq-K}.

Unfortunately, we could not include results from the method of \cite{cikm2012}. This algorithm
submits queries to Web search engines to i) enrich the product titles with important missing words
(one query per title), and ii) to assign importance scores to the words of the enriched titles (one
query \textit{per word pair, per title}). If we applied this method on the \textit{Aggregate}
dataset of Skroutz (about $24 \cdot 10^4$ titles and 7 words per title), the required number of
queries would be 5.3 million. Clearly, this cost renders the method entirely unsustainable.

Moreover, notice that in \cite{cikm2012}, the proposed method is compared against only one
similarity metric by employing only 2 small datasets. In contrast, here we evaluate UPM and UPM+
against 4 similarity metrics by using 18 datasets.

The experiments were conducted on a machine with an Intel CoreI7 7700@3.6GHz CPU and 32GB
of RAM, running Ubuntu Linux 16.04 LTS. All methods were implemented in C++ and compiled by gcc
with the -O3 speed optimization flag. We have made both this code and the datasets publicly
available on GitHub\footnote{https://github.com/lakritidis/UPM} to allow the interested
researchers verify our results and work further on our findings.

\begin{table}[t]
\begin{center}
\small
\caption{The experimental datasets accompanied by their characteristics}
\label{tab-data}
\begin{tabular}{|l|c|c|c|c|}\hline
{\bf Dataset}          &  $|V|$  &  $|P|$  & {\bf Titles} &  $\overline{l_t}$ \\\hline
CPUs                   &    37   &  1901   &     3862     &       11.285      \\\hline
Digital Cameras        &   103   &   836   &     2697     &        9.605      \\\hline
Dishwashers            &    94   &  1678   &     3424     &        6.819      \\\hline
Microwaves             &   114   &  1039   &     2342     &        7.591      \\\hline
Mobile Phones          &    84   &  1837   &     4081     &        8.416      \\\hline
Refrigerators          &   118   &  5172   &    11291     &        7.847      \\\hline
TVs                    &   129   &  1678   &     3564     &       10.263      \\\hline
Washing Machines       &    87   &  1703   &     4044     &        7.931      \\\hline
PriceRunner Aggregate  &   306   & 15844   &    35305     &        8.560      \\\hline\hline
Air Conditioners       &   216   &  1442   &    13595     &       10.497      \\\hline
Car Batteries          &    66   &  2097   &     5864     &        8.073      \\\hline
Cookers \& Ovens       &   163   &  1355   &    10858     &        6.455      \\\hline
CPUs                   &    92   &   356   &     1906     &        9.115      \\\hline
Digital Cameras        &   152   &   973   &     4111     &        8.802      \\\hline
Refrigerators          &   161   &  1697   &    16177     &        5.955      \\\hline
TVs                    &   205   &  1246   &     7002     &        7.382      \\\hline
Watches                &   212   & 60559   &   178657     &        6.517      \\\hline
Skroutz Aggregate      &   652   & 68512   &   238170     &        6.827      \\\hline
\end{tabular}
\end{center}
\vspace{-0.2cm}
\end{table}

\subsection{Effectiveness Evaluation}
\label{ssec-effec}

The experimentation process is organized into two phases: In this subsection we study the
effectiveness of the proposed algorithm, whereas in Subsection \ref{ssec-effic} we examine its
efficiency. In both phases, the five parameters of the algorithm are fixed according to the
discussion of Subsection \ref{ssec-tune}.

UPM and UPM+ achieve product matching by generating clusters of similar products. To evaluate
their output we applied the following methodology: Initially, we iterate through each cluster and
for each product in the cluster, we create one pairwise match record with each of the rest of the
products in the same cluster. In other words, we create a database with all the distinct product
pairs within a cluster. In the sequel, we compare the records of this database with the ones of the
aforementioned matches file and we count the number of true positives and negatives.

The matching quality was measured by employing the $F1$ score, defined by
$F1=2\mathcal{PR}/\mathcal{(P+R)}$, where $\mathcal{P}$ and $\mathcal{R}$ represent the values of
Precision and the Recall, respectively.

Figure \ref{fig-f1-pr} illustrates the performance of UPM and UPM+ against the aforementioned
methods, for the 9 datasets of PriceRunner. Each diagram depicts the fluctuation of the $F1$ scores
for various similarity thresholds ranging from 0.1 to 0.9. Recall that the similarity threshold
$\tau$ determines whether two entities $e_1$ and $e_2$ match or not. That is, $e_1$ matches $e_2$
only if their similarity value exceeds $\tau$. Since in Subsection \ref{ssec-tune} we fixed
$\tau=0.4$, the $F1$ scores of UPM and UPM+ are represented by horizontal lines.

The first conclusion which derives from these diagrams is that in all datasets, the similarity
metrics $\mathcal{CS'}$ and $\mathcal{J'}$ with IDF token weights performed much better than their
standard expressions $\mathcal{CS}$ and $\mathcal{J}$. For instance, in the \textit{Aggregate}
dataset of Fig. \ref{fig-f1-pr}i, the effectiveness of $\mathcal{CS'}$ was $F1=0.39$; compared to
the corresponding $F1$ scores of $\mathcal{CS}$ and $\mathcal{J}$, this value was higher by 221\%
and 214\% respectively. Similar differences are also observed for $\mathcal{J'}$: its matching
quality surpassed that of $\mathcal{CS}$ and $\mathcal{J}$ by 219\% and 211\%. For this reason, we
omit the commentation of the plain cosine similarity $\mathcal{CS}$ and Jaccard index $\mathcal{J}$
in the discussion which follows.

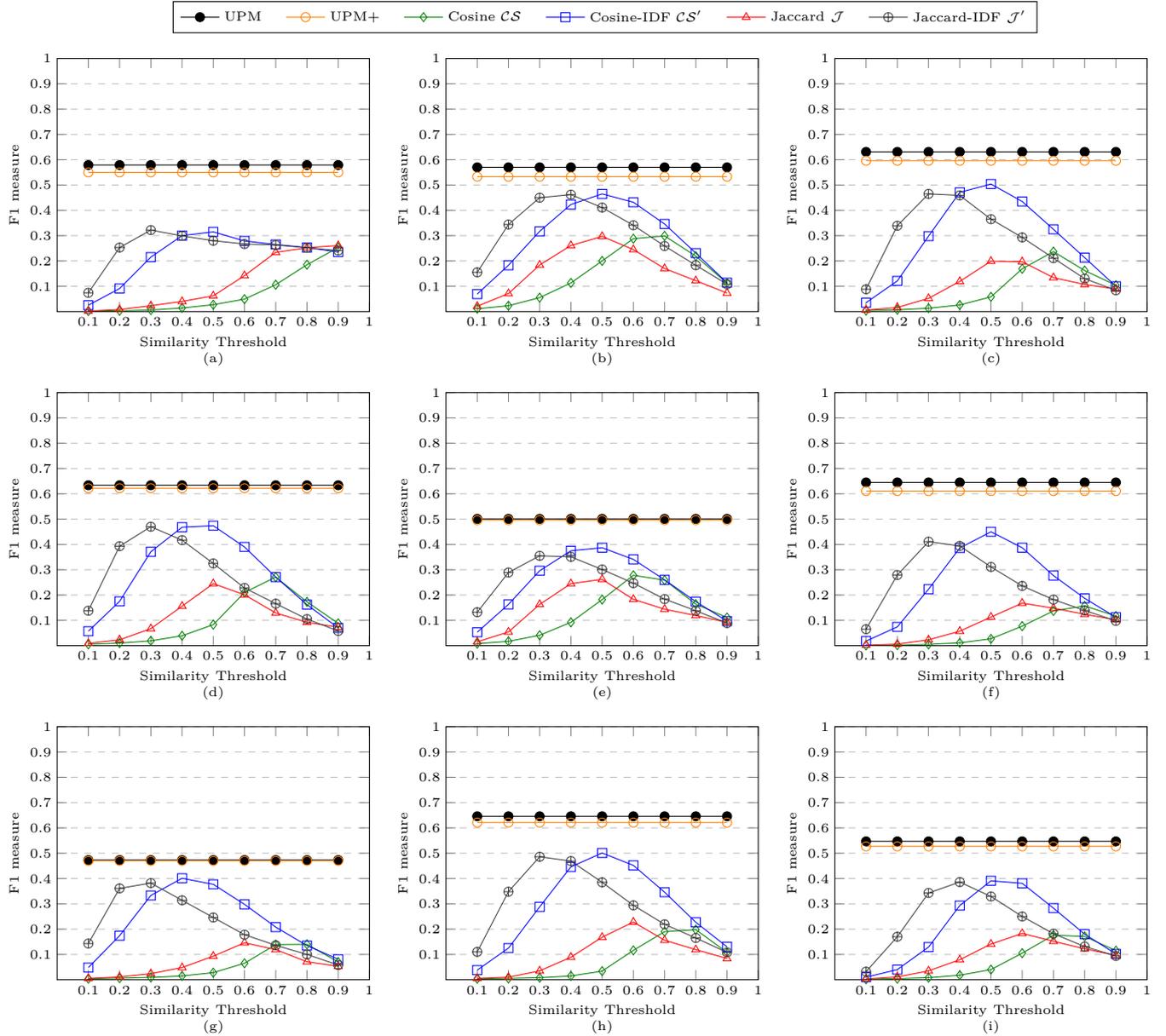
\begin{figure*}[!htbp]
\begin{tikzpicture}
\smaller
    \begin{axis}[
        xlabel={Similarity Threshold\\(a)},
        xlabel style={align=center},
        ylabel={F1 measure},
        xmin=0, xmax=1,
        ymin=0, ymax=1,
        xtick={0.1,0.2,0.3,0.4,0.5,0.6,0.7,0.8,0.9,1.0},
        ytick={0.1,0.2,0.3,0.4,0.5,0.6,0.7,0.8,0.9,1.0},
        ymajorgrids=true,
        grid style=dashed,
        legend style={
           at={(0.37,1.1)},anchor=south west,legend columns=-1,  /tikz/every even column/.append style={column sep=0.3cm}
        }
    ]
    \addplot[color=black, mark=*] coordinates { (0.1,0.579)(0.2,0.579)(0.3,0.579)(0.4,0.579)(0.5,0.579)(0.6,0.579)(0.7,0.579)(0.8,0.579)(0.9,0.579) };
    \addplot[color=orange, mark=o] coordinates { (0.1,0.550)(0.2,0.550)(0.3,0.550)(0.4,0.550)(0.5,0.550)(0.6,0.550)(0.7,0.550)(0.8,0.550)(0.9,0.550) };
    \addplot[color=green!50!black, mark=diamond] coordinates { (0.1,0.002)(0.2,0.003)(0.3,0.006)(0.4,0.014)(0.5,0.027)(0.6,0.049)(0.7,0.106)(0.8,0.185)(0.9,0.253) };
    \addplot[color=blue, mark=square] coordinates { (0.1,0.025)(0.2,0.091)(0.3,0.215)(0.4,0.300)(0.5,0.315)(0.6,0.279)(0.7,0.265)(0.8,0.253)(0.9,0.235) };
    \addplot[color=red, mark=triangle] coordinates { (0.1,0.002)(0.2,0.008)(0.3,0.023)(0.4,0.040)(0.5,0.062)(0.6,0.142)(0.7,0.234)(0.8,0.252)(0.9,0.261) };
    \addplot[color=gray!50!black, mark=oplus] coordinates { (0.1,0.074)(0.2,0.253)(0.3,0.322)(0.4,0.299)(0.5,0.280)(0.6,0.267)(0.7,0.263)(0.8,0.252)(0.9,0.241) };
    \smaller
    \legend{UPM, UPM+, Cosine $\mathcal{CS}$, Cosine-IDF $\mathcal{CS'}$, Jaccard $\mathcal{J}$, Jaccard-IDF $\mathcal{J'}$}
    \end{axis}

    \begin{axis}[
        xshift=6cm,
        xlabel={Similarity Threshold\\(b)},
        xlabel style={align=center},
        ylabel={F1 measure},
        xmin=0, xmax=1,
        ymin=0, ymax=1,
        xtick={0.1,0.2,0.3,0.4,0.5,0.6,0.7,0.8,0.9,1.0},
        ytick={0.1,0.2,0.3,0.4,0.5,0.6,0.7,0.8,0.9,1.0},
        ymajorgrids=true,
        grid style=dashed
    ]
    \addplot[color=black, mark=*] coordinates { (0.1,0.570)(0.2,0.570)(0.3,0.570)(0.4,0.570)(0.5,0.570)(0.6,0.570)(0.7,0.570)(0.8,0.570)(0.9,0.570) };
    \addplot[color=orange, mark=o] coordinates { (0.1,0.533)(0.2,0.533)(0.3,0.533)(0.4,0.533)(0.5,0.533)(0.6,0.533)(0.7,0.533)(0.8,0.533)(0.9,0.533) };
    \addplot[color=green!50!black, mark=diamond] coordinates { (0.1,0.011)(0.2,0.023)(0.3,0.055)(0.4,0.113)(0.5,0.200)(0.6,0.288)(0.7,0.299)(0.8,0.219)(0.9,0.112) };
    \addplot[color=blue, mark=square] coordinates { (0.1,0.069)(0.2,0.183)(0.3,0.317)(0.4,0.423)(0.5,0.465)(0.6,0.432)(0.7,0.346)(0.8,0.231)(0.9,0.114) };
    \addplot[color=red, mark=triangle] coordinates { (0.1,0.021)(0.2,0.071)(0.3,0.184)(0.4,0.261)(0.5,0.297)(0.6,0.245)(0.7,0.170)(0.8,0.122)(0.9,0.073) };
    \addplot[color=gray!50!black, mark=oplus] coordinates { (0.1,0.155)(0.2,0.344)(0.3,0.450)(0.4,0.462)(0.5,0.411)(0.6,0.341)(0.7,0.259)(0.8,0.183)(0.9,0.108) };
    \smaller
    \end{axis}

    \begin{axis}[
        xshift=12cm,
        xlabel={Similarity Threshold\\(c)},
        xlabel style={align=center},
        ylabel={F1 measure},
        xmin=0, xmax=1,
        ymin=0, ymax=1,
        xtick={0.1,0.2,0.3,0.4,0.5,0.6,0.7,0.8,0.9,1.0},
        ytick={0.1,0.2,0.3,0.4,0.5,0.6,0.7,0.8,0.9,1.0},
        ymajorgrids=true,
        grid style=dashed
    ]
    \addplot[color=black, mark=*] coordinates { (0.1,0.631)(0.2,0.631)(0.3,0.631)(0.4,0.631)(0.5,0.631)(0.6,0.631)(0.7,0.631)(0.8,0.631)(0.9,0.631) };
    \addplot[color=orange, mark=o] coordinates { (0.1,0.596)(0.2,0.596)(0.3,0.596)(0.4,0.596)(0.5,0.596)(0.6,0.596)(0.7,0.596)(0.8,0.596)(0.9,0.596) };
    \addplot[color=green!50!black, mark=diamond] coordinates { (0.1,0.004)(0.2,0.007)(0.3,0.013)(0.4,0.026)(0.5,0.058)(0.6,0.169)(0.7,0.237)(0.8,0.162)(0.9,0.104) };
    \addplot[color=blue, mark=square] coordinates { (0.1,0.035)(0.2,0.121)(0.3,0.298)(0.4,0.471)(0.5,0.504)(0.6,0.435)(0.7,0.325)(0.8,0.213)(0.9,0.100) };
    \addplot[color=red, mark=triangle] coordinates { (0.1,0.006)(0.2,0.016)(0.3,0.052)(0.4,0.119)(0.5,0.199)(0.6,0.197)(0.7,0.134)(0.8,0.107)(0.9,0.089) };
    \addplot[color=gray!50!black, mark=oplus] coordinates { (0.1,0.088)(0.2,0.339)(0.3,0.465)(0.4,0.459)(0.5,0.365)(0.6,0.293)(0.7,0.211)(0.8,0.130)(0.9,0.084) };
    \smaller
    \end{axis}

    \begin{axis}[
        yshift=-5.2cm,
        xlabel={Similarity Threshold\\(d)},
        xlabel style={align=center},
        ylabel={F1 measure},
        xmin=0, xmax=1,
        ymin=0, ymax=1,
        xtick={0.1,0.2,0.3,0.4,0.5,0.6,0.7,0.8,0.9,1.0},
        ytick={0.1,0.2,0.3,0.4,0.5,0.6,0.7,0.8,0.9,1.0},
        ymajorgrids=true,
        grid style=dashed
    ]
    \addplot[color=black, mark=*] coordinates { (0.1,0.634)(0.2,0.634)(0.3,0.634)(0.4,0.634)(0.5,0.634)(0.6,0.634)(0.7,0.634)(0.8,0.634)(0.9,0.634) };
    \addplot[color=orange, mark=o] coordinates { (0.1,0.622)(0.2,0.622)(0.3,0.622)(0.4,0.622)(0.5,0.622)(0.6,0.622)(0.7,0.622)(0.8,0.622)(0.9,0.622) };
    \addplot[color=green!50!black, mark=diamond] coordinates { (0.1,0.005)(0.2,0.010)(0.3,0.019)(0.4,0.039)(0.5,0.083)(0.6,0.210)(0.7,0.271)(0.8,0.173)(0.9,0.088) };
    \addplot[color=blue, mark=square] coordinates { (0.1,0.057)(0.2,0.175)(0.3,0.371)(0.4,0.468)(0.5,0.474)(0.6,0.390)(0.7,0.271)(0.8,0.162)(0.9,0.070) };
    \addplot[color=red, mark=triangle] coordinates { (0.1,0.009)(0.2,0.023)(0.3,0.067)(0.4,0.156)(0.5,0.245)(0.6,0.201)(0.7,0.129)(0.8,0.092)(0.9,0.073) };
    \addplot[color=gray!50!black, mark=oplus] coordinates { (0.1,0.138)(0.2,0.393)(0.3,0.470)(0.4,0.417)(0.5,0.325)(0.6,0.228)(0.7,0.166)(0.8,0.105)(0.9,0.058) };
    \smaller
    \end{axis}

    \begin{axis}[
        yshift=-5.2cm,
        xshift=6cm,
        xlabel={Similarity Threshold\\(e)},
        xlabel style={align=center},
        ylabel={F1 measure},
        xmin=0, xmax=1,
        ymin=0, ymax=1,
        xtick={0.1,0.2,0.3,0.4,0.5,0.6,0.7,0.8,0.9,1.0},
        ytick={0.1,0.2,0.3,0.4,0.5,0.6,0.7,0.8,0.9,1.0},
        ymajorgrids=true,
        grid style=dashed,
    ]
    \addplot[color=black, mark=*] coordinates { (0.1,0.501)(0.2,0.501)(0.3,0.501)(0.4,0.501)(0.5,0.501)(0.6,0.501)(0.7,0.501)(0.8,0.501)(0.9,0.501) };
    \addplot[color=orange, mark=o] coordinates { (0.1,0.496)(0.2,0.496)(0.3,0.496)(0.4,0.496)(0.5,0.496)(0.6,0.496)(0.7,0.496)(0.8,0.496)(0.9,0.496) };
    \addplot[color=green!50!black, mark=diamond] coordinates { (0.1,0.008)(0.2,0.017)(0.3,0.041)(0.4,0.092)(0.5,0.182)(0.6,0.278)(0.7,0.259)(0.8,0.164)(0.9,0.111) };
    \addplot[color=blue, mark=square] coordinates { (0.1,0.053)(0.2,0.163)(0.3,0.296)(0.4,0.375)(0.5,0.387)(0.6,0.341)(0.7,0.260)(0.8,0.174)(0.9,0.097) };
    \addplot[color=red, mark=triangle] coordinates { (0.1,0.015)(0.2,0.054)(0.3,0.163)(0.4,0.245)(0.5,0.262)(0.6,0.183)(0.7,0.144)(0.8,0.119)(0.9,0.092) };
    \addplot[color=gray!50!black, mark=oplus] coordinates { (0.1,0.132)(0.2,0.289)(0.3,0.355)(0.4,0.351)(0.5,0.301)(0.6,0.247)(0.7,0.184)(0.8,0.138)(0.9,0.089) };
    \smaller
    \end{axis}

    \begin{axis}[
        yshift=-5.2cm,
        xshift=12cm,
        xlabel={Similarity Threshold\\(f)},
        xlabel style={align=center},
        ylabel={F1 measure},
        xmin=0, xmax=1,
        ymin=0, ymax=1,
        xtick={0.1,0.2,0.3,0.4,0.5,0.6,0.7,0.8,0.9,1.0},
        ytick={0.1,0.2,0.3,0.4,0.5,0.6,0.7,0.8,0.9,1.0},
        ymajorgrids=true,
        grid style=dashed,
    ]
    \addplot[color=black, mark=*] coordinates { (0.1,0.645)(0.2,0.645)(0.3,0.645)(0.4,0.645)(0.5,0.645)(0.6,0.645)(0.7,0.645)(0.8,0.645)(0.9,0.645) };
    \addplot[color=orange, mark=o] coordinates { (0.1,0.611)(0.2,0.611)(0.3,0.611)(0.4,0.611)(0.5,0.611)(0.6,0.611)(0.7,0.611)(0.8,0.611)(0.9,0.611) };
    \addplot[color=green!50!black, mark=diamond] coordinates { (0.1,0.001)(0.2,0.002)(0.3,0.005)(0.4,0.011)(0.5,0.027)(0.6,0.077)(0.7,0.137)(0.8,0.157)(0.9,0.117) };
    \addplot[color=blue, mark=square] coordinates { (0.1,0.018)(0.2,0.074)(0.3,0.223)(0.4,0.386)(0.5,0.450)(0.6,0.387)(0.7,0.277)(0.8,0.187)(0.9,0.112) };
    \addplot[color=red, mark=triangle] coordinates { (0.1,0.002)(0.2,0.006)(0.3,0.023)(0.4,0.057)(0.5,0.113)(0.6,0.168)(0.7,0.148)(0.8,0.124)(0.9,0.102) };
    \addplot[color=gray!50!black, mark=oplus] coordinates { (0.1,0.065)(0.2,0.279)(0.3,0.411)(0.4,0.394)(0.5,0.311)(0.6,0.236)(0.7,0.182)(0.8,0.139)(0.9,0.098) };
    \smaller
    \end{axis}

    \begin{axis}[
        yshift=-10.4cm,
        xlabel={Similarity Threshold\\(g)},
        xlabel style={align=center},
        ylabel={F1 measure},
        xmin=0, xmax=1,
        ymin=0, ymax=1,
        xtick={0.1,0.2,0.3,0.4,0.5,0.6,0.7,0.8,0.9,1.0},
        ytick={0.1,0.2,0.3,0.4,0.5,0.6,0.7,0.8,0.9,1.0},
        ymajorgrids=true,
        grid style=dashed,
    ]
    \addplot[color=black, mark=*] coordinates { (0.1,0.473)(0.2,0.473)(0.3,0.473)(0.4,0.473)(0.5,0.473)(0.6,0.473)(0.7,0.473)(0.8,0.473)(0.9,0.473) };
    \addplot[color=orange, mark=o] coordinates { (0.1,0.470)(0.2,0.470)(0.3,0.470)(0.4,0.470)(0.5,0.470)(0.6,0.470)(0.7,0.470)(0.8,0.470)(0.9,0.470) };
    \addplot[color=green!50!black, mark=diamond] coordinates { (0.1,0.004)(0.2,0.006)(0.3,0.009)(0.4,0.015)(0.5,0.028)(0.6,0.065)(0.7,0.139)(0.8,0.140)(0.9,0.070) };
    \addplot[color=blue, mark=square] coordinates { (0.1,0.048)(0.2,0.174)(0.3,0.333)(0.4,0.401)(0.5,0.377)(0.6,0.298)(0.7,0.208)(0.8,0.135)(0.9,0.081) };
    \addplot[color=red, mark=triangle] coordinates { (0.1,0.005)(0.2,0.011)(0.3,0.024)(0.4,0.048)(0.5,0.093)(0.6,0.146)(0.7,0.119)(0.8,0.071)(0.9,0.053) };
    \addplot[color=gray!50!black, mark=oplus] coordinates { (0.1,0.143)(0.2,0.361)(0.3,0.382)(0.4,0.314)(0.5,0.246)(0.6,0.178)(0.7,0.136)(0.8,0.100)(0.9,0.058) };
    \smaller
    \end{axis}

    \begin{axis}[
        yshift=-10.4cm,
        xshift=6cm,
        xlabel={Similarity Threshold\\(h)},
        xlabel style={align=center},
        ylabel={F1 measure},
        xmin=0, xmax=1,
        ymin=0, ymax=1,
        xtick={0.1,0.2,0.3,0.4,0.5,0.6,0.7,0.8,0.9,1.0},
        ytick={0.1,0.2,0.3,0.4,0.5,0.6,0.7,0.8,0.9,1.0},
        ymajorgrids=true,
        grid style=dashed,
    ]
    \addplot[color=black, mark=*] coordinates { (0.1,0.646)(0.2,0.646)(0.3,0.646)(0.4,0.646)(0.5,0.646)(0.6,0.646)(0.7,0.646)(0.8,0.646)(0.9,0.646) };
    \addplot[color=orange, mark=o] coordinates { (0.1,0.621)(0.2,0.621)(0.3,0.621)(0.4,0.621)(0.5,0.621)(0.6,0.621)(0.7,0.621)(0.8,0.621)(0.9,0.621) };
    \addplot[color=green!50!black, mark=diamond] coordinates { (0.1,0.004)(0.2,0.005)(0.3,0.008)(0.4,0.015)(0.5,0.034)(0.6,0.116)(0.7,0.190)(0.8,0.198)(0.9,0.108) };
    \addplot[color=blue, mark=square] coordinates { (0.1,0.038)(0.2,0.125)(0.3,0.288)(0.4,0.446)(0.5,0.501)(0.6,0.452)(0.7,0.346)(0.8,0.227)(0.9,0.130) };
    \addplot[color=red, mark=triangle] coordinates { (0.1,0.005)(0.2,0.010)(0.3,0.035)(0.4,0.089)(0.5,0.168)(0.6,0.228)(0.7,0.156)(0.8,0.119)(0.9,0.084) };
    \addplot[color=gray!50!black, mark=oplus] coordinates { (0.1,0.110)(0.2,0.348)(0.3,0.486)(0.4,0.469)(0.5,0.385)(0.6,0.294)(0.7,0.219)(0.8,0.165)(0.9,0.108) };
    \smaller
    \end{axis}

    \begin{axis}[
        yshift=-10.4cm,
        xshift=12cm,
        xlabel={Similarity Threshold\\(i)},
        xlabel style={align=center},
        ylabel={F1 measure},
        xmin=0, xmax=1,
        ymin=0, ymax=1,
        xtick={0.1,0.2,0.3,0.4,0.5,0.6,0.7,0.8,0.9,1.0},
        ytick={0.1,0.2,0.3,0.4,0.5,0.6,0.7,0.8,0.9,1.0},
        ymajorgrids=true,
        grid style=dashed,
    ]
    \addplot[color=black, mark=*] coordinates { (0.1,0.547)(0.2,0.547)(0.3,0.547)(0.4,0.547)(0.5,0.547)(0.6,0.547)(0.7,0.547)(0.8,0.547)(0.9,0.547) };
    \addplot[color=orange, mark=o] coordinates { (0.1,0.527)(0.2,0.527)(0.3,0.527)(0.4,0.527)(0.5,0.527)(0.6,0.527)(0.7,0.527)(0.8,0.527)(0.9,0.527) };
    \addplot[color=green!50!black, mark=diamond] coordinates { (0.1,0.002)(0.2,0.004)(0.3,0.008)(0.4,0.018)(0.5,0.040)(0.6,0.105)(0.7,0.176)(0.8,0.171)(0.9,0.115) };
    \addplot[color=blue, mark=square] coordinates { (0.1,0.010)(0.2,0.040)(0.3,0.129)(0.4,0.293)(0.5,0.391)(0.6,0.381)(0.7,0.283)(0.8,0.180)(0.9,0.102) };
    \addplot[color=red, mark=triangle] coordinates { (0.1,0.003)(0.2,0.011)(0.3,0.035)(0.4,0.079)(0.5,0.141)(0.6,0.183)(0.7,0.152)(0.8,0.122)(0.9,0.096) };
    \addplot[color=gray!50!black, mark=oplus] coordinates { (0.1,0.032)(0.2,0.170)(0.3,0.343)(0.4,0.386)(0.5,0.329)(0.6,0.250)(0.7,0.182)(0.8,0.131)(0.9,0.094) };
    \smaller
    \end{axis}
\end{tikzpicture}
\caption{Performance comparison of various product matching methods for the PriceRunner datasets:
a) CPUs, b) digital cameras, c) dishwashers, d) microwaves, e) mobile phones, f) refrigerators, g)
TVs, h) washing machines, and i) aggregate dataset.}
\label{fig-f1-pr}
\vspace{-0.25cm}
\end{figure*}

UPM prevailed over its adversary approaches in all cases. The highest $F1$ values were observed
in the cases of \textit{Washing Machines} (0.646), \textit{Refrigerators} (0.645),
\textit{Microwave Ovens} (0.634), and \textit{Dishwashers} (0.631) in Figures
\ref{fig-f1-pr}h, \ref{fig-f1-pr}f, \ref{fig-f1-pr}d, and \ref{fig-f1-pr}c, respectively. The
strongest opponent was $\mathcal{CS'}$, as its $F1$ scores were 0.501 (-22\%), 0.45 (-30\%), 0.474
(-25\%), and 0.504 (-20\%) for the aforementioned datasets, respectively. The largest percentage
difference was measured in the case of \textit{CPUs} (Fig \ref{fig-f1-pr}a), where our method
achieved an $F1$ value which was about 84\% greater than the respective $F1$ value of
$\mathcal{CS'}$. On the other hand, the smallest difference was observed in the \textit{TVs}
dataset, and it was roughly 17\%.

Furthermore, UPM won $\mathcal{CS'}$ in the large and heterogeneous \textit{Aggregate} dataset
(Fig. \ref{fig-f1-pr}i), since its $F1$ was 0.547, compared to the value of 0.391 which was
achieved by the latter method (that is, approximately 40\% higher). Apart from the \textit{CPUs}
dataset, the results of Jaccard index $\mathcal{J'}$ were slightly worse than those of
$\mathcal{CS'}$, consequently, UPM outperformed this metric by an even greater margin.

Regarding UPM+, in most cases, its effectiveness was very close to the one of UPM. Recall that
UPM+ attempts to improve the execution time of the algorithm by processing only the first $2K^{*}$
tokens of a product title, and by pruning the rest of them. For the datasets which contained
\textit{Mobile Phones} (Fig. \ref{fig-f1-pr}e) and \textit{TVs} (Fig. \ref{fig-f1-pr}g), the
two algorithms performed almost equally well. Additionally, for the \textit{Aggregate} dataset,
UPM is only 3.6\% more accurate than UPM+. These measurements verify the theoretical
foundation which predicted that only a portion of the words of the titles are important for the
identification of a product.

\begin{figure*}
\begin{tikzpicture}
\smaller
    \begin{axis}[
        xlabel={Similarity Threshold\\(a)},
        xlabel style={align=center},
        ylabel={F1 measure},
        xmin=0, xmax=1,
        ymin=0, ymax=1,
        xtick={0.1,0.2,0.3,0.4,0.5,0.6,0.7,0.8,0.9,1.0},
        ytick={0.1,0.2,0.3,0.4,0.5,0.6,0.7,0.8,0.9,1.0},
        ymajorgrids=true,
        grid style=dashed,
        legend style={
                at={(0.37,1.1)},anchor=south west,legend columns=-1,  /tikz/every even column/.append style={column sep=0.3cm}
        }
    ]
    \addplot[color=black, mark=*] coordinates { (0.1,0.673)(0.2,0.673)(0.3,0.673)(0.4,0.673)(0.5,0.673)(0.6,0.673)(0.7,0.673)(0.8,0.673)(0.9,0.673) };
    \addplot[color=orange, mark=o] coordinates { (0.1,0.667)(0.2,0.667)(0.3,0.667)(0.4,0.667)(0.5,0.667)(0.6,0.667)(0.7,0.667)(0.8,0.667)(0.9,0.667) };
    \addplot[color=green!50!black, mark=diamond] coordinates { (0.1,0.015)(0.2,0.030)(0.3,0.068)(0.4,0.147)(0.5,0.268)(0.6,0.372)(0.7,0.314)(0.8,0.170)(0.9,0.067) };
    \addplot[color=blue, mark=square] coordinates { (0.1,0.107)(0.2,0.259)(0.3,0.377)(0.4,0.448)(0.5,0.474)(0.6,0.443)(0.7,0.357)(0.8,0.221)(0.9,0.080) };
    \addplot[color=red, mark=triangle] coordinates { (0.1,0.027)(0.2,0.089)(0.3,0.238)(0.4,0.336)(0.5,0.323)(0.6,0.203)(0.7,0.124)(0.8,0.080)(0.9,0.039) };
    \addplot[color=gray!50!black, mark=oplus] coordinates { (0.1,0.225)(0.2,0.386)(0.3,0.458)(0.4,0.453)(0.5,0.406)(0.6,0.324)(0.7,0.223)(0.8,0.136)(0.9,0.061) };
    \smaller
    \legend{UPM, UPM+, Cosine $\mathcal{CS}$, Cosine-IDF $\mathcal{CS'}$, Jaccard $\mathcal{J}$, Jaccard-IDF $\mathcal{J'}$}
    \end{axis}

    \begin{axis}[
        xshift=6cm,
        xlabel={Similarity Threshold\\(b)},
        xlabel style={align=center},
        ylabel={F1 measure},
        xmin=0, xmax=1,
        ymin=0, ymax=1,
        xtick={0.1,0.2,0.3,0.4,0.5,0.6,0.7,0.8,0.9,1.0},
        ytick={0.1,0.2,0.3,0.4,0.5,0.6,0.7,0.8,0.9,1.0},
        ymajorgrids=true,
        grid style=dashed,
    ]
    \addplot[color=black, mark=*] coordinates { (0.1,0.719)(0.2,0.719)(0.3,0.719)(0.4,0.719)(0.5,0.719)(0.6,0.719)(0.7,0.719)(0.8,0.719)(0.9,0.719) };
    \addplot[color=orange, mark=o] coordinates { (0.1,0.676)(0.2,0.676)(0.3,0.676)(0.4,0.676)(0.5,0.676)(0.6,0.676)(0.7,0.676)(0.8,0.676)(0.9,0.676) };
    \addplot[color=green!50!black, mark=diamond] coordinates { (0.1,0.003)(0.2,0.006)(0.3,0.015)(0.4,0.035)(0.5,0.113)(0.6,0.287)(0.7,0.376)(0.8,0.251)(0.9,0.110) };
    \addplot[color=blue, mark=square] coordinates { (0.1,0.066)(0.2,0.235)(0.3,0.538)(0.4,0.637)(0.5,0.585)(0.6,0.431)(0.7,0.254)(0.8,0.124)(0.9,0.032) };
    \addplot[color=red, mark=triangle] coordinates { (0.1,0.005)(0.2,0.018)(0.3,0.099)(0.4,0.226)(0.5,0.328)(0.6,0.311)(0.7,0.194)(0.8,0.115)(0.9,0.060) };
    \addplot[color=gray!50!black, mark=oplus] coordinates { (0.1,0.155)(0.2,0.499)(0.3,0.604)(0.4,0.535)(0.5,0.425)(0.6,0.271)(0.7,0.154)(0.8,0.086)(0.9,0.041) };
    \smaller
    \end{axis}

    \begin{axis}[
        xshift=12cm,
        xlabel={Similarity Threshold\\(c)},
        xlabel style={align=center},
        ylabel={F1 measure},
        xmin=0, xmax=1,
        ymin=0, ymax=1,
        xtick={0.1,0.2,0.3,0.4,0.5,0.6,0.7,0.8,0.9,1.0},
        ytick={0.1,0.2,0.3,0.4,0.5,0.6,0.7,0.8,0.9,1.0},
        ymajorgrids=true,
        grid style=dashed
    ]
    \addplot[color=black, mark=*] coordinates { (0.1,0.761)(0.2,0.761)(0.3,0.761)(0.4,0.761)(0.5,0.761)(0.6,0.761)(0.7,0.761)(0.8,0.761)(0.9,0.761) };
    \addplot[color=orange, mark=o] coordinates { (0.1,0.732)(0.2,0.732)(0.3,0.732)(0.4,0.732)(0.5,0.732)(0.6,0.732)(0.7,0.732)(0.8,0.732)(0.9,0.732) };
    \addplot[color=green!50!black, mark=diamond] coordinates { (0.1,0.015)(0.2,0.028)(0.3,0.054)(0.4,0.101)(0.5,0.188)(0.6,0.375)(0.7,0.415)(0.8,0.294)(0.9,0.173) };
    \addplot[color=blue, mark=square] coordinates { (0.1,0.094)(0.2,0.293)(0.3,0.443)(0.4,0.533)(0.5,0.532)(0.6,0.462)(0.7,0.359)(0.8,0.233)(0.9,0.114) };
    \addplot[color=red, mark=triangle] coordinates { (0.1,0.025)(0.2,0.064)(0.3,0.175)(0.4,0.305)(0.5,0.370)(0.6,0.316)(0.7,0.244)(0.8,0.201)(0.9,0.131) };
    \addplot[color=gray!50!black, mark=oplus] coordinates { (0.1,0.194)(0.2,0.428)(0.3,0.520)(0.4,0.500)(0.5,0.436)(0.6,0.353)(0.7,0.268)(0.8,0.192)(0.9,0.103) };
    \smaller
    \end{axis}

    \begin{axis}[
       	yshift=-5.2cm,
        xlabel={Similarity Threshold\\(d)},
        xlabel style={align=center},
        ylabel={F1 measure},
        xmin=0, xmax=1,
        ymin=0, ymax=1,
        xtick={0.1,0.2,0.3,0.4,0.5,0.6,0.7,0.8,0.9,1.0},
        ytick={0.1,0.2,0.3,0.4,0.5,0.6,0.7,0.8,0.9,1.0},
        ymajorgrids=true,
        grid style=dashed
    ]
    \addplot[color=black, mark=*] coordinates { (0.1,0.932)(0.2,0.932)(0.3,0.932)(0.4,0.932)(0.5,0.932)(0.6,0.932)(0.7,0.932)(0.8,0.932)(0.9,0.932) };
    \addplot[color=orange, mark=o] coordinates { (0.1,0.823)(0.2,0.823)(0.3,0.823)(0.4,0.823)(0.5,0.823)(0.6,0.823)(0.7,0.823)(0.8,0.823)(0.9,0.823) };
    \addplot[color=green!50!black, mark=diamond] coordinates { (0.1,0.027)(0.2,0.050)(0.3,0.108)(0.4,0.220)(0.5,0.357)(0.6,0.444)(0.7,0.334)(0.8,0.178)(0.9,0.080) };
    \addplot[color=blue, mark=square] coordinates { (0.1,0.266)(0.2,0.594)(0.3,0.651)(0.4,0.519)(0.5,0.336)(0.6,0.172)(0.7,0.090)(0.8,0.039)(0.9,0.018) };
    \addplot[color=red, mark=triangle] coordinates { (0.1,0.045)(0.2,0.142)(0.3,0.353)(0.4,0.407)(0.5,0.356)(0.6,0.213)(0.7,0.140)(0.8,0.099)(0.9,0.044) };
    \addplot[color=gray!50!black, mark=oplus] coordinates { (0.1,0.509)(0.2,0.606)(0.3,0.415)(0.4,0.262)(0.5,0.160)(0.6,0.088)(0.7,0.065)(0.8,0.049)(0.9,0.039) };
    \smaller
    \end{axis}

    \begin{axis}[
        yshift=-5.2cm,
    	xshift=6cm,
        xlabel={Similarity Threshold\\(e)},
        xlabel style={align=center},
        ylabel={F1 measure},
        xmin=0, xmax=1,
        ymin=0, ymax=1,
        xtick={0.1,0.2,0.3,0.4,0.5,0.6,0.7,0.8,0.9,1.0},
        ytick={0.1,0.2,0.3,0.4,0.5,0.6,0.7,0.8,0.9,1.0},
        ymajorgrids=true,
        grid style=dashed
    ]
    \addplot[color=black, mark=*] coordinates { (0.1,0.603)(0.2,0.603)(0.3,0.603)(0.4,0.603)(0.5,0.603)(0.6,0.603)(0.7,0.603)(0.8,0.603)(0.9,0.603) };
    \addplot[color=orange, mark=o] coordinates { (0.1,0.580)(0.2,0.580)(0.3,0.580)(0.4,0.580)(0.5,0.580)(0.6,0.580)(0.7,0.580)(0.8,0.580)(0.9,0.580) };
    \addplot[color=green!50!black, mark=diamond] coordinates { (0.1,0.027)(0.2,0.058)(0.3,0.116)(0.4,0.203)(0.5,0.312)(0.6,0.397)(0.7,0.345)(0.8,0.221)(0.9,0.096) };
    \addplot[color=blue, mark=square] coordinates { (0.1,0.114)(0.2,0.284)(0.3,0.436)(0.4,0.513)(0.5,0.488)(0.6,0.389)(0.7,0.267)(0.8,0.149)(0.9,0.048) };
    \addplot[color=red, mark=triangle] coordinates { (0.1,0.052)(0.2,0.144)(0.3,0.296)(0.4,0.362)(0.5,0.360)(0.6,0.246)(0.7,0.167)(0.8,0.114)(0.9,0.072) };
    \addplot[color=gray!50!black, mark=oplus] coordinates { (0.1,0.236)(0.2,0.454)(0.3,0.514)(0.4,0.458)(0.5,0.265)(0.6,0.282)(0.7,0.201)(0.8,0.119)(0.9,0.057) };
    \smaller
    \end{axis}

    \begin{axis}[
        yshift=-5.2cm,
        xshift=12cm,
        xlabel={Similarity Threshold\\(f)},
        xlabel style={align=center},
        ylabel={F1 measure},
        xmin=0, xmax=1,
        ymin=0, ymax=1,
        xtick={0.1,0.2,0.3,0.4,0.5,0.6,0.7,0.8,0.9,1.0},
        ytick={0.1,0.2,0.3,0.4,0.5,0.6,0.7,0.8,0.9,1.0},
        ymajorgrids=true,
        grid style=dashed,
    ]
    \addplot[color=black, mark=*] coordinates { (0.1,0.851)(0.2,0.851)(0.3,0.851)(0.4,0.851)(0.5,0.851)(0.6,0.851)(0.7,0.851)(0.8,0.851)(0.9,0.851) };
    \addplot[color=orange, mark=o] coordinates { (0.1,0.793)(0.2,0.793)(0.3,0.793)(0.4,0.793)(0.5,0.793)(0.6,0.793)(0.7,0.793)(0.8,0.793)(0.9,0.793) };
    \addplot[color=green!50!black, mark=diamond] coordinates { (0.1,0.010)(0.2,0.018)(0.3,0.041)(0.4,0.081)(0.5,0.163)(0.6,0.337)(0.7,0.422)(0.8,0.345)(0.9,0.180) };1186
    \addplot[color=blue, mark=square] coordinates { (0.1,0.070)(0.2,0.268)(0.3,0.461)(0.4,0.561)(0.5,0.568)(0.6,0.490)(0.7,0.372)(0.8,0.248)(0.9,0.139) };
    \addplot[color=red, mark=triangle] coordinates { (0.1,0.016)(0.2,0.049)(0.3,0.142)(0.4,0.271)(0.5,0.339)(0.6,0.355)(0.7,0.269)(0.8,0.215)(0.9,0.149) };
    \addplot[color=gray!50!black, mark=oplus] coordinates { (0.1,0.153)(0.2,0.447)(0.3,0.546)(0.4,0.516)(0.5,0.438)(0.6,0.354)(0.7,0.277)(0.8,0.205)(0.9,0.137) };
    \smaller
    \end{axis}

    \begin{axis}[
        yshift=-10.4cm,
        xlabel={Similarity Threshold\\(g)},
        xlabel style={align=center},
        ylabel={F1 measure},
        xmin=0, xmax=1,
        ymin=0, ymax=1,
        xtick={0.1,0.2,0.3,0.4,0.5,0.6,0.7,0.8,0.9,1.0},
        ytick={0.1,0.2,0.3,0.4,0.5,0.6,0.7,0.8,0.9,1.0},
        ymajorgrids=true,
        grid style=dashed,
    ]
    \addplot[color=black, mark=*] coordinates { (0.1,0.710)(0.2,0.710)(0.3,0.710)(0.4,0.710)(0.5,0.710)(0.6,0.710)(0.7,0.710)(0.8,0.710)(0.9,0.710) };
    \addplot[color=orange, mark=o] coordinates { (0.1,0.625)(0.2,0.625)(0.3,0.625)(0.4,0.625)(0.5,0.625)(0.6,0.625)(0.7,0.625)(0.8,0.625)(0.9,0.625) };
    \addplot[color=green!50!black, mark=diamond] coordinates { (0.1,0.010)(0.2,0.014)(0.3,0.024)(0.4,0.042)(0.5,0.075)(0.6,0.147)(0.7,0.191)(0.8,0.153)(0.9,0.085) };
    \addplot[color=blue, mark=square] coordinates { (0.1,0.069)(0.2,0.218)(0.3,0.424)(0.4,0.510)(0.5,0.480)(0.6,0.398)(0.7,0.299)(0.8,0.205)(0.9,0.100) };
    \addplot[color=red, mark=triangle] coordinates { (0.1,0.013)(0.2,0.028)(0.3,0.061)(0.4,0.111)(0.5,0.151)(0.6,0.161)(0.7,0.119)(0.8,0.091)(0.9,0.070) };
    \addplot[color=gray!50!black, mark=oplus] coordinates { (0.1,0.160)(0.2,0.376)(0.3,0.463)(0.4,0.414)(0.5,0.340)(0.6,0.275)(0.7,0.208)(0.8,0.144)(0.9,0.085) };
    \smaller
    \end{axis}

    \begin{axis}[
        yshift=-10.4cm,
        xshift=6cm,
        xlabel={Similarity Threshold\\(h)},
        xlabel style={align=center},
        ylabel={F1 measure},
        xmin=0, xmax=1,
        ymin=0, ymax=1,
        xtick={0.1,0.2,0.3,0.4,0.5,0.6,0.7,0.8,0.9,1.0},
        ytick={0.1,0.2,0.3,0.4,0.5,0.6,0.7,0.8,0.9,1.0},
        ymajorgrids=true,
        grid style=dashed,
    ]
    \addplot[color=black, mark=*] coordinates {(0.1,0.558)(0.2,0.558)(0.3,0.558)(0.4,0.558)(0.5,0.558)(0.6,0.558)(0.7,0.558)(0.8,0.558)(0.9,0.558)};
    \addplot[color=orange, mark=o] coordinates { (0.1,0.295)(0.2,0.295)(0.3,0.295)(0.4,0.295)(0.5,0.295)(0.6,0.295)(0.7,0.295)(0.8,0.295)(0.9,0.295) };
    \addplot[color=green!50!black, mark=diamond] coordinates { (0.1,0.001)(0.2,0.001)(0.3,0.002)(0.4,0.004)(0.5,0.012)(0.6,0.042)(0.7,0.104)(0.8,0.200)(0.9,0.178) };
    \addplot[color=blue, mark=square] coordinates { (0.1,0.005)(0.2,0.022)(0.3,0.080)(0.4,0.181)(0.5,0.276)(0.6,0.314)(0.7,0.249)(0.8,0.137)(0.9,0.049) };
    \addplot[color=red, mark=triangle] coordinates { (0.1,0.002)(0.2,0.005)(0.3,0.007)(0.4,0.025)(0.5,0.054)(0.6,0.134)(0.7,0.184)(0.8,0.202)(0.9,0.153) };
    \addplot[color=gray!50!black, mark=oplus] coordinates { (0.1,0.010)(0.2,0.050)(0.3,0.138)(0.4,0.218)(0.5,0.245)(0.6,0.222)(0.7,0.156)(0.8,0.092)(0.9,0.040) };
    \smaller
    \end{axis}

    \begin{axis}[
        yshift=-10.4cm,
        xshift=12cm,
        xlabel={Similarity Threshold\\(i)},
        xlabel style={align=center},
        ylabel={F1 measure},
        xmin=0, xmax=1,
        ymin=0, ymax=1,
        xtick={0.1,0.2,0.3,0.4,0.5,0.6,0.7,0.8,0.9,1.0},
        ytick={0.1,0.2,0.3,0.4,0.5,0.6,0.7,0.8,0.9,1.0},
        ymajorgrids=true,
        grid style=dashed,
    ]
    \addplot[color=black, mark=*] coordinates { (0.1,0.649)(0.2,0.649)(0.3,0.649)(0.4,0.649)(0.5,0.649)(0.6,0.649)(0.7,0.649)(0.8,0.649)(0.9,0.649) };
    \addplot[color=orange, mark=o] coordinates { (0.1,0.481)(0.2,0.481)(0.3,0.481)(0.4,0.481)(0.5,0.481)(0.6,0.481)(0.7,0.481)(0.8,0.481)(0.9,0.481) };
    \addplot[color=green!50!black, mark=diamond] coordinates { (0.1,0.0)(0.2,0.0)(0.3,0.004)(0.4,0.009)(0.5,0.026)(0.6,0.084)(0.7,0.172)(0.8,0.220)(0.9,0.145) };
    \addplot[color=blue, mark=square] coordinates { (0.1,0.008)(0.2,0.039)(0.3,0.127)(0.4,0.268)(0.5,0.378)(0.6,0.393)(0.7,0.298)(0.8,0.161)(0.9,0.051) };
    \addplot[color=red, mark=triangle] coordinates { (0.1,0.0)(0.2,0.005)(0.3,0.016)(0.4,0.054)(0.5,0.102)(0.6,0.187)(0.7,0.190)(0.8,0.168)(0.9,0.115) };
    \addplot[color=gray!50!black, mark=oplus] coordinates { (0.1,0.049)(0.2,0.089)(0.3,0.218)(0.4,0.313)(0.5,0.321)(0.6,0.267)(0.7,0.182)(0.8,0.108)(0.9,0.046) };
    \smaller
    \end{axis}
\end{tikzpicture}
\caption{Performance comparison of various product matching methods for the Skroutz datasets: a)
air conditioners, b) car batteries, c) cookers \& ovens, d) CPUs, e) digital cameras, f)
refrigerators, g) TVs, h) watches, and i) aggregate dataset. }
\label{fig-f1-sk}
\vspace{-0.25cm}
\end{figure*}
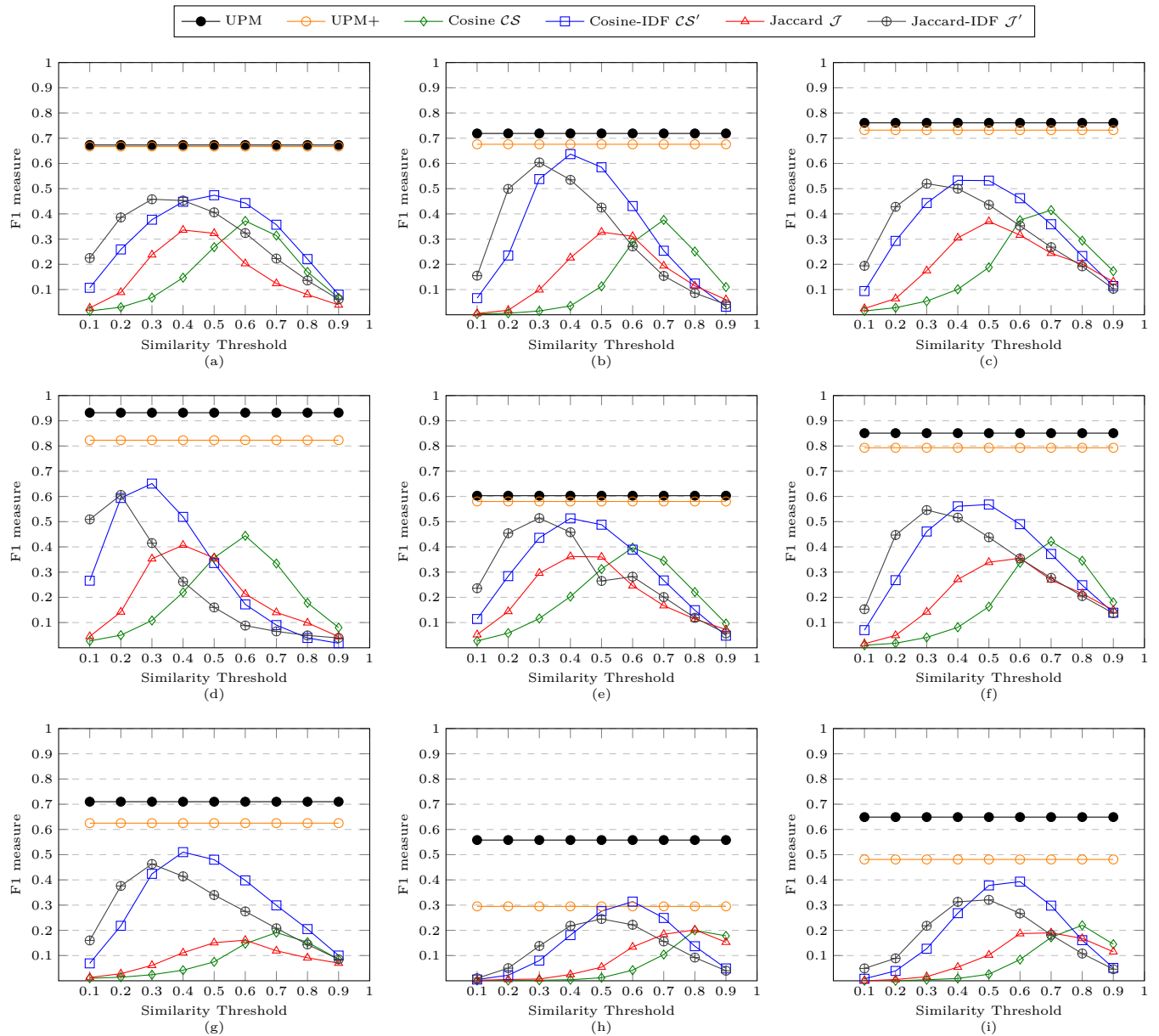

The results indicate the superiority of UPM and UPM+ against their adversary methods in
multiple types of products, and in the heterogeneous \textit{Aggregate} dataset. The situation is
improved even further in the datasets which originate from Skroutz (Fig. \ref{fig-f1-sk}).
In two cases, namely \textit{CPUs} (Fig. \ref{fig-f1-sk}d), and \textit{Refrigerators} (Fig.
\ref{fig-f1-sk}f), UPM approached 100\% precision, with $F1$ being equal to 0.932 and 0.851
respectively. In the same datasets, the effectiveness of $\mathcal{CS'}$ was 0.651 (-30\%) and
0.568 (-33\%) respectively, whereas $\mathcal{J'}$ achieved 0.606 (-35\%) for \textit{CPUs}, and
0.546 (-36\%) for \textit{Refrigerators}. Moreover, in the \textit{Aggregate} dataset, UPM outperformed $\mathcal{CS'}$ and $\mathcal{J'}$
by nearly 71\% and 104\% respectively.

Unlike the previous case, the performance of UPM+ was not so stable compared to UPM. In some
datasets the two methods achieve product matching of almost equal quality, such as \textit{Air
Conditioners} (Fig. \ref{fig-f1-sk}a) and \textit{Digital Cameras} (Fig. \ref{fig-f1-sk}e).
However, there are occasions where the difference is larger, like the cases of \textit{Watches}
(Fig. \ref{fig-f1-sk}h) and \textit{Aggregate} datasets (Fig. \ref{fig-f1-sk}i). Here UPM+ is
inferior to UPM, by 53\% and 26\% respectively. \textit{Watches} is the only case where UPM+
is defeated by cosine similarity, even marginally, by 6\%.

Apart from their superior effectiveness, the proposed algorithms are also parameter-free, whereas
the performance of the pairwise matching methods depends heavily on the selected similarity
threshold value. In particular, the maximum effectiveness of $\mathcal{CS'}$ was observed for four
different values of $\tau$: $\tau=0.3$ (e.g. \textit{CPUs} in Fig. \ref{fig-f1-sk}d), $\tau=0.4$
(e.g. \textit{Car Batteries} in Fig. \ref{fig-f1-sk}b), $\tau=0.5$ (e.g. \textit{Air Conditioners}
in Fig. \ref{fig-f1-sk}a), and $\tau=0.6$ (e.g. \textit{Watches} in Fig. \ref{fig-f1-sk}h). A
similar observation is also valid for the Jaccard index.

Four datasets, that is, \textit{CPUs, Digital Cameras, Refrigerators}, and \textit{TVs} have been
crawled from both product comparison platforms. The examination of all methods on these datasets
leads to the conclusion that the effectiveness does not primarily depend on the category itself.
Instead, it is rather affected by how accurately the vendors describe their products. For instance,
the $F1$ score of UPM for the \textit{CPUs} of PriceRunner and Skroutz was 0.579 and 0.932
respectively, a difference of about 61\%. On the contrary, this difference was only 6\% for the
\textit{Digital Cameras}. Moreover, UPM performed better on the \textit{Refrigerators} rather
than the \textit{CPUs} of PriceRunner, whereas the opposite occurred on the corresponding datasets
of Skroutz.

\subsection{Efficiency Evaluation}
\label{ssec-effic}

This subsection contains the experimental measurements of the efficiency of the proposed algorithm
in comparison with the aforementioned pairwise matching methods. In summary, the reported results
demonstrate that in contrast to the other methods, both UPM and UPM+ are fast enough to be
applied to all datasets, even to the larger ones.

Figure~\ref{fig-eff} depicts the running times (in seconds) of the six unsupervised product
matching methods which participate in our evaluation. More specifically, the two diagrams
illustrate the duration of the execution of these methods in the 9 datasets of PriceRunner (top
diagram), and the 9 datasets of Skroutz (bottom diagram). The vertical axis of time is in
logarithmic scale, to reliably display the large time differences between these executions.

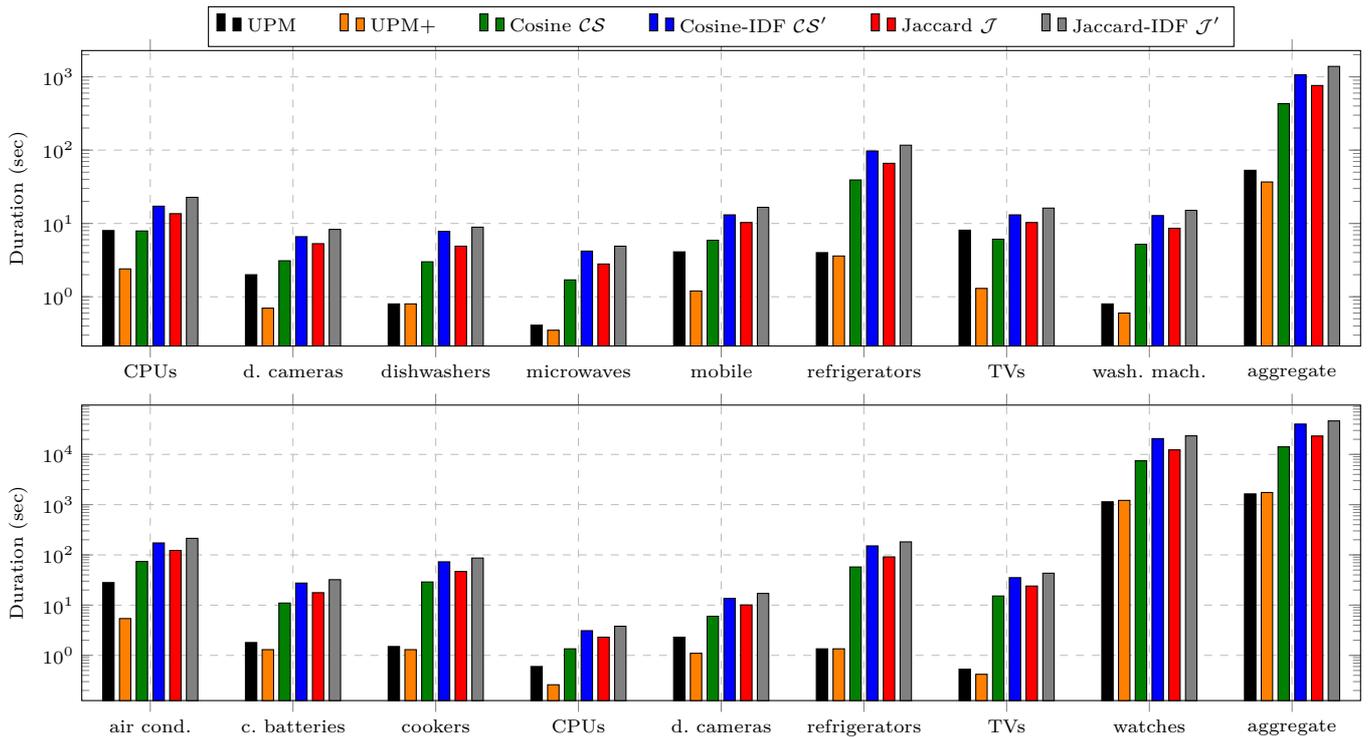
\begin{figure*}
\begin{tikzpicture}
\smaller
\begin{axis}[
    ybar,ymode = log, log origin=infty,
    height=5.5cm, width=18.4cm,
    bar width=0.15cm,
    enlargelimits=0.06,
    legend style={at={(0.5,1.15)}, anchor=north,legend columns=-1, /tikz/every even column/.append style={column sep=0.5cm}},
    ylabel={Duration (sec)},
    symbolic x coords={{CPUs},{d. cameras},{dishwashers},{microwaves},{mobile},{refrigerators},{TVs},{wash. mach.},{aggregate}},
    xtick=data,
    ymajorgrids=true,
    xmajorgrids=true,
    grid style=dashed,
    ]
\addplot [fill=black] coordinates  {({CPUs},8.0) ({d. cameras},2) ({dishwashers},0.80) ({microwaves},0.41) ({mobile},4.1) ({refrigerators},4) ({TVs},8.1) ({wash. mach.},0.8) ({aggregate},53.1)};

\addplot [fill=orange] coordinates {({CPUs},2.4) ({d. cameras},0.7) ({dishwashers},0.8) ({microwaves},0.35) ({mobile},1.2) ({refrigerators},3.6) ({TVs},1.3) ({wash. mach.},0.6) ({aggregate},36.8)};

\addplot [fill=green!50!black] coordinates {({CPUs},7.9) ({d. cameras},3.1) ({dishwashers},3) ({microwaves},1.7) ({mobile},5.9) ({refrigerators},39.2) ({TVs},6.1) ({wash. mach.},5.2) ({aggregate},430.3)};

\addplot [fill=blue] coordinates {({CPUs},17.2) ({d. cameras},6.6) ({dishwashers},7.8) ({microwaves},4.2) ({mobile},13.1) ({refrigerators},97.3) ({TVs},13.1) ({wash. mach.},12.8) ({aggregate},1064.9)};

\addplot [fill=red] coordinates {({CPUs},13.6) ({d. cameras},5.3) ({dishwashers},4.9) ({microwaves},2.8) ({mobile},10.3) ({refrigerators},66.1) ({TVs},10.3) ({wash. mach.},8.6) ({aggregate},761.2)};

\addplot [fill=gray] coordinates {({CPUs},22.7) ({d. cameras},8.3) ({dishwashers},8.9) ({microwaves},4.9) ({mobile},16.6) ({refrigerators},116.7) ({TVs},16.2) ({wash. mach.},15.1) ({aggregate},1386.9)};

\legend{UPM, UPM+, Cosine $\mathcal{CS}$, Cosine-IDF $\mathcal{CS'}$, Jaccard $\mathcal{J}$, Jaccard-IDF $\mathcal{J'}$}
\end{axis}

\begin{axis}[
    ybar,ymode = log, log origin=infty,
    yshift=-4.7cm,
    height=5.5cm, width=18.4cm,
    bar width=0.15cm,
    enlargelimits=0.06,
    legend style={at={(0.5,1.15)}, anchor=north,legend columns=-1, /tikz/every even column/.append style={column sep=0.5cm}},
    ylabel={Duration (sec)},
    symbolic x coords={{air cond.},{c. batteries},{cookers},{CPUs},{d. cameras},{refrigerators},{TVs},{watches},{aggregate}},
    xtick=data,
    ymajorgrids=true,
    xmajorgrids=true,
    grid style=dashed,
    ]
\addplot [fill=black] coordinates {({air cond.},28.1) ({c. batteries},1.8) ({cookers},1.5) ({CPUs},0.6) ({d. cameras},2.3) ({refrigerators},1.35) ({TVs},0.53) ({watches},1141) ({aggregate},1638)};

\addplot [fill=orange] coordinates {({air cond.},5.4) ({c. batteries},1.3) ({cookers},1.3) ({CPUs},0.26) ({d. cameras},1.1) ({refrigerators},1.35) ({TVs},0.42) ({watches},1216) ({aggregate},1744)};

\addplot [fill=green!50!black] coordinates {({air cond.},74) ({c. batteries},11) ({cookers},28.9) ({CPUs},1.35) ({d. cameras},6) ({refrigerators},57.5) ({TVs},15.2) ({watches},7470) ({aggregate},14156)};

\addplot [fill=blue] coordinates {({air cond.},173) ({c. batteries},27.5) ({cookers},72.8) ({CPUs},3.1) ({d. cameras},13.6) ({refrigerators},151.1) ({TVs},35.2) ({watches},20642) ({aggregate},40176)};

\addplot [fill=red] coordinates {({air cond.},122.7) ({c. batteries},17.7) ({cookers},46.7) ({CPUs},2.3) ({d. cameras},10.1) ({refrigerators},91.2) ({TVs},24) ({watches},12329) ({aggregate},23257)};

\addplot [fill=gray] coordinates {({air cond.},213.5) ({c. batteries},32.1) ({cookers},86.5) ({CPUs},3.8) ({d. cameras},17.1) ({refrigerators},181.5) ({TVs},43.2) ({watches},23452) ({aggregate},46442)};

\end{axis}
\end{tikzpicture}
\caption{Running times (in seconds) of various product matching methods for the i) PriceRunner
(top), and ii) Skroutz (bottom) datasets. The vertical axis is in logarithmic scale, whereas the
legend on the top of the Figure is common for both diagrams.}
\label{fig-eff}
\end{figure*}

Regarding the PriceRunner datasets, UPM+ was the fastest method among its adversaries, whereas
the basic method, UPM, was ranked second. Notice that the larger the value of $K$ is, the greater
the performance gap becomes. This is anticipated, since a high $K$ value leads to a big number of
combinations to be extracted and scored. For instance, the average title length of
\textit{Dishwashers} was 7.591 (Table \ref{tab-data}), therefore, we set $K=3$ according to eq.
\ref{eq-K}. For such small values of $K$, UPM and UPM+ were equally fast (0.8 sec). On the
other hand, for \textit{CPUs}, where $K$ was equal to 5, UPM+ was more than 3 times faster than
UPM (2.4 vs 8 sec). Similarly, for \textit{TVs} where $K$ was also 5, UPM+ was 6.2 times
faster than UPM (1.3 vs 8.1 sec).

\begin{table}[b]
\begin{center}
\small
\vspace{-0.3cm}
\caption{Efficiency comparison of various methods on the \textit{Aggregate} datasets}
\label{tab-agg-ds}
\begin{tabular}{|l|c|c|c|c|} \hline
\multirow{2}{*}{\bf Method} & \multicolumn{2}{c|}{\bf PriceRunner Aggregate} & \multicolumn{2}{c|}{\bf Skroutz Aggregate} \\\cline{2-5}
                            & Time (sec) & Gain (x) & Time (sec) &  Gain (x)  \\\hline
UPM+                        &      37    &    -     &   1638    &     -      \\\hline
UPM                         &      53    &   1.43   &   1744    &     1.06   \\\hline
$\mathcal{CS}$              &     430    &  11.62   &  14156    &     8.64   \\\hline
$\mathcal{CS'}$             &    1065    &  28.78   &  40176    &    24.52   \\\hline
$\mathcal{J}$               &     761    &  20.56   &  23257    &    14.20   \\\hline
$\mathcal{J'}$              &    1387    &  37.48   &  46442    &    28.35   \\\hline
\end{tabular}
\end{center}
\end{table}

Both UPM and UPM+ were substantially faster than the plain string similarity metrics. Notice
that the larger the dataset is, the higher the performance gap becomes, due to the quadratic
complexity of the pairwise matching procedure. The slowest methods were the ones which were the
strongest opponents in terms of matching quality, that is, $\mathcal{CS'}$ and $\mathcal{J'}$. For
instance, in the \textit{CPUs} dataset UPM was 2.2 and 2.8 times faster than $\mathcal{CS'}$ and
$\mathcal{J'}$ respectively, whereas, UPM+ outperformed these metrics by 7.2 and 9.5 times. We
will shortly discuss the \textit{Aggregate} dataset.

The efficiency measurements were also positive for our proposed algorithms in the datasets which
originated from Skroutz. Hence, in the case of \textit{Watches}, UPM and UPM+ consumed equal
times, and they were faster than $\mathcal{CS}$, $\mathcal{CS'}$, $\mathcal{J}$ and $\mathcal{J'}$
by 6.5, 18.1, 10.8, and 20.5 times respectively. Remarkably, in some datasets such as \textit{TVs},
\textit{Refrigerators}, and \textit{Cookers \& Ovens}, our algorithms were faster than the pairwise
methods by more or less than two orders of magnitude.

Finally, Table \ref{tab-agg-ds} presents the execution times and the efficiency differences of
the six examined methods on the two \textit{Aggregate} datasets from PriceRunner and Skroutz.
According to Table \ref{tab-data} and eq. \ref{eq-K}, the first dataset was processed with
$K=4$, whereas the second with $K=3$. Consequently, UPM+ achieved better times in the first case
and it was faster than $\mathcal{CS}$, $\mathcal{CS'}$, $\mathcal{J}$ and $\mathcal{J'}$ by roughly
11.6, 28.8, 20.6 and 37.5 times respectively. The corresponding performance gaps in the
\textit{Aggregate} dataset of Skroutz were also very high, approximating 8.6, 24.5, 14.2 and 28.4
times respectively.

\section{Conclusions}
\label{sec-conclusion}

In this paper we introduced UPM, a clustering-based unsupervised algorithm for matching product
titles from different data sources. This problem is particularly important for the e-commerce
industry since it facilitates the comparison of product features and prices. UPM implements
multiple novel elements, the most important of which are:

\begin{itemize}
\item{it does not perform pairwise comparison of the titles, thus, it avoids the quadratic
complexity of this procedure. Instead, it achieves matching by groupping the titles of identical
products into clusters,}
\item{it partially identifies the semantics of the title words,}
\item{it includes a post-processing verification stage which corrects the erroneous matchings by
moving products through clusters and by creating new clusters.}
\end{itemize}

In addition, we introduced UPM+, a variant which prunes the titles and processes only a
portion of their words.

The exhaustive experimental evaluation of UPM and UPM+ on 18 datasets from two product
comparison platforms demonstrated their superiority over the traditional pairwise matching methods.
More specifically in terms of matching quality, our method outperformed 4 similarity metrics by a
margin of up to 84\%. Furthermore, it was about 24--37 times faster than the pairwise matching
methods in large datasets. In some cases, the performance was improved by more than two
orders of magnitude.


\bibliographystyle{spmpsci}
\bibliography{air}

\end{document}